\documentclass[draft, onecolumn]{IEEEtran}

\usepackage[utf8]{inputenc} 
\usepackage[T1]{fontenc}
\usepackage{url}
\usepackage{ifthen}
\usepackage{cite}
\usepackage[cmex10]{amsmath} % Use the [cmex10] option to ensure complicance
\usepackage{amsmath,amsthm}
\usepackage{tikz}
\usepackage{bm}
\usepackage{mathdots}
\usepackage{pgfplots}
\pgfplotsset{compat=1.15}
\usepackage{mathrsfs}
\usetikzlibrary{arrows}
\usepackage{yhmath}
\usepackage{cancel}
\usepackage{color}
\usepackage{siunitx}
\usepackage{array}

\usepackage{multirow}
\usepackage[noend]{algpseudocode}
\usepackage{hyperref}
\usepackage{mathtools}
\usepackage{amssymb}
\usepackage{algorithm}
\usepackage{caption}
\usepackage{comment}
\usepackage{subcaption}
\usepackage{gensymb}
\usepackage{tabularx}
\usepackage{mathtools}

\usepackage{setspace}

\theoremstyle{definition}
\newtheorem{definition}{Definition}
\newtheorem{remark}{Remark}
\newtheorem{example}{Example}
\newtheorem*{example*}{Example}
\newtheorem{theorem}{Theorem}

\newtheorem{lemma}{Lemma}

\usepackage{arydshln}

\title{Vers: fully distributed Coded Computing System with Distributed Encoding} 
\author{%
  Nastaran Abadi Khooshemehr, Mohammad Ali Maddah-Ali \\
  Sharif University of Technology,  University of Minnesota \\
  Email: abadi.nastaran@ee.sharif.ir,  maddah@umn.edu
}
\begin{document}

% %%% Single author, or several authors with same affiliation:
% \author{%
%   \IEEEauthorblockN{Stefan M.~Moser}
%   \IEEEauthorblockA{ETH Zürich\\
%                     ISI (D-ITET)\\
%                     CH-8092 Zürich, Switzerland\\
%                     Email: moser@isi.ee.ethz.ch}
% }

%%% Several authors with up to three affiliations:

\maketitle

\section{Abstract}
Coded computing has proved to be useful in distributed computing, and has addressed challenges such as straggler workers. We have observed that almost all coded computing systems studied so far consider a setup of one master and some workers. However, recently emerging technologies such as blockchain, internet of things, and federated learning introduce new requirements for coded computing systems. In these systems, data is generated (and probably stored) in a distributed manner, so central encoding/decoding by a master is not feasible and scalable. This paper presents a fully distributed coded computing system that consists of $k\in\mathbb{N}$ data owners and $N\in\mathbb{N}$ workers, where data owners employ workers to do some computations on their data, as specified by a target function $f$ of degree $d\in\mathbb{N}$. As there is no central encoder, workers perform encoding themselves, prior to computation phase.
The challenge in this system is the presence of adversarial data owners that do not know the data of honest data owners but cause discrepancies by sending different versions of data to different workers, which is detrimental to local encodings in workers. There are at most $\beta\in\mathbb{N}$ adversarial data owners, and each distributes at most $v\in\mathbb{N}$ different versions of data. Since the adversaries and their possibly colluded behavior are not known to workers and honest data owners, workers compute tags of their received data, in addition to their main computational task, and send them to data owners in order to help them in decoding. We introduce a tag function that allows data owners to partition workers into sets that previously had received the same data from all data owners. Then, we characterize the fundamental limit of this fully distributed coded computing system, denoted by $t^*$, which is the minimum number of workers whose work can be used to correctly calculate the desired function of data of honest data owners. We show that $t^*=v^{\beta}d(K-1)+1$, and present converse and achievable proofs.

\section{Introduction}\label{introduction section}
Coded computing utilizes coding theory techniques to address the fundamental challenges in distributed computing systems such as straggler mitigation, security, and privacy preservation, and communication bandwidth issues \cite{li2020coded}. The core idea in coded computing is to inject data or computation redundancy in order to accomplish the mentioned goals. For instance, coded computing can be used to accomplish a computation task by using the work of faster worker nodes, and not being dependent on any specific set of workers (straggler mitigation). This line of research has received so much attention, and many aspects of it have been explored, see \cite{ng2021comprehensive} for a survey on coded distributed computing. However, almost all coded computing researches consider a setup of one master (also called other names like parameter server) and several workers, where the master encodes the data, distributes the data and computation tasks among workers, collects the partial results from workers, and constructs the final result. The following reasons motivate us to pursue a different path from the centralized, single-master coded computing systems.
\begin{itemize}
    \item In a distributed computing system where a single master encodes the whole initial data and then decodes all of the results, the master would be a bottleneck for scalability. As the number of workers grows, the master should get more powerful, in terms of computation and communication.
    \item
    In some applications like federated learning \cite{yang2019federated}, the data for computation tasks are not congregated in one place, for reasons like the privacy of the data owners, and/or the huge volume of data. Thus, one master does not have access to the whole data.
    \item
    Computation schemes that are specifically designed for distributed sources of data are better suited for applications where the production of data is distributed by nature, like sharded blockchains, federated learning, and internet of things.
    \item
    In practice, a cluster of servers that act as workers, would serve more than one master, and do computation tasks for several masters in parallel. In this way, the resources of workers would be efficiently utilized, and not left idle.
\end{itemize}
A few works have studied \textit{masterless} coded computing or coded computing with several masters. In \cite{jeong2020fully}, a fully-decentralized and masterless setup, consisting of only $N\in\mathbb{N}$ workers, has been considered. In such a setup, decentralized coding strategies for two particular computation tasks, matrix multiplication, and fast Fourier transform, have been introduced. In matrix multiplication of two arbitrary matrices $\mathbf{A}$ and $\mathbf{B}$, each worker has a portion of $\mathbf{A}$ and $\mathbf{B}$ initially. Then, each worker performs local encoding on its initial data, communicates with other workers, and then multiplies two matrices. At the end of the algorithm, each worker has a portion of the calculated $\mathbf{AB}$. Note that the encodings done by workers, are on their own initial data, and not on data received from other workers. We will argue that distributed encoding of data received from external sources introduces new challenges.

In \cite{sun2021coded}, several masters, each with a matrix-vector multiplication task, share some workers, and the workers have heterogeneous computing powers. The problem is to assign computation tasks to workers such that the overall completion time is minimum. Each worker either serves one master or several masters and divides its computation power between them. In the latter case, each worker simply performs the computation tasks of the masters separately, and no coding is involved. In \cite{wang2022all}, there are $N\in\mathbb{N}$ workers, where each has a piece of data, and wants to have a linear combination of all data. There is no master and workers exchange messages with one another in a communication-efficient way, to calculate their target linear combinations. \cite{sundaram2008distributed} considers a network graph of $N\in\mathbb{N}$ nodes each with an initial value, where a subset of them wants to calculate some arbitrary functions of the initial values. The target function of each node in that subset could be different. In order to achieve this goal, an iterative linear strategy is deployed in which nodes receive messages from their neighbours and obtain a linear combination of the received values and their own value, without any leader. Given the graph of nodes has some particular connectivity properties, a finite number of rounds suffices for nodes to calculate the desired functions of initial values, using the linear combinations obtained in the previous iterations. A similar problem with the presence of adversaries is studied in \cite{sundaram2010distributed}, where adversaries may calculate the linear combinations incorrectly and send incorrect values to their neighbours. The adversaries do not send different values to different neighbours, and may only deviate in calculating the linear combination of the messages of their neighbours and their own message.

Coded computing can be designed to guarantee security in an adversarial environment, as in \cite{yang2021coded,soleymani2021list,chen2018draco,yu2019lagrange}. In the conventional setup of a single master and many workers, the security of the coded computing scheme translates to the resilience of the scheme against adversarial workers, because in such setups, the master is assumed to be honest. Adversarial workers can freely deviate from the protocol, e.g. return an arbitrary result, send different data to different workers, refuse to send data, and so forth. However, when there are several masters, it is reasonable to assume that some masters might be adversarial as well.

In this paper, we study a coded computing system comprised of several masters and workers, where masters employ workers to accomplish computational jobs for them. We refer to masters as \textit{data owners}, as each of them has a piece of data, and aims to have a polynomial target function of that data by using workers.
Some of the data owners are adversarial and try to create discrepancies in the system, in order to cause errors in the computation tasks of honest data owners.  There are communication links between data owners and workers, and also between adversarial data owners, since they can collude with one another. However, like many coded computing systems \cite{li2020coded,ng2021comprehensive}, there is no communication means between workers in our system.
The system works as follows. The data owners distribute their data among workers, but adversarial data owners may distribute different contradictory data to different workers, to corrupt the results. The adversarial data owners are free to cooperate in choosing what they send to workers.  Upon the reception of data, each worker encodes the received data of all data owners, applies the target function on the encoded data, and then returns the result back to data owners, along with a small \textit{tag} of the received data, to inform the honest data owners about the adversarial data owners indirectly. Tags allow honest data owners to partition workers into sets who have received the same message from the adversaries in the first step, and enables honest data owners to decode the results. Finally, the honest data owners can extract their required information from the returned results of the workers. 

We name the described system \textit{Vers}, which is short for versatile. The reason for this naming is that the workers are indeed versatile, and do a range of different tasks: they encode the received data, calculate the target function of the coded data, and calculate tags. We study the fundamental limit of Vers in the case where Lagrange encoding is deployed in workers. This fundamental limit, which we denote by $t^*$, is the minimum number such that honest data owners can correctly and reliably extract their required information from the results of any set of $t^*$ workers. In other words, for any adversarial behavior, any set of $t^*$ workers should be enough for honest data owners to calculate the target function of their data correctly. On the other hand, for $t^*$ to be the minimum indeed, there should exist an adversarial behavior and a set of $t^*-1$ workers whose results do not determine the target of honest data owners uniquely.
It is worth noting that in our previous work \cite{khooshemehr2021fundamental}, we have studied the fundamental limit of a system in which some data owners want to store their data in some storage nodes, such that data can be correctly retrieved later. In that system, some data owners are adversarial and may send different data to different storage nodes.

The type of adversarial behavior in our problem, i.e. adversaries sending inconsistent data to other nodes, has been explored from different points of views in a wide range of distributed systems, such as byzantine agreement \cite{pease1980reaching,lamport2019byzantine,canetti1993fast}, byzantine broadcast \cite{bracha1984asynchronous,cachin2001secure,miller2016honey}, verifiable information dispersal \cite{cachin2005asynchronous,hendricks2007verifying,yang2022dispersedledger}, and distributed key generation \cite{abraham2021reaching}. Even though the adversarial behaviour has been previously known, addressing such behaviour in the context of distributed computing systems like blockchain, federated learning and internet of things is very underexplored. The distinction between this work and the mentioned ones that consider a similar adversarial model is that we deal with computation, rather than consensus. Our goal is to correctly calculate functions of a subset of data, and this does not necessarily require consensus on that data before calculation. 
Generally, inconsistencies caused by the adversaries need to be resolved somehow, for which there are two main approaches: nodes can either resolve the inconsistencies between themselves by communicating with one another, or they can postpone dealing with inconsistencies to subsequent steps such as decoding. The first approach is pursued in many works such as \cite{cachin2005asynchronous}, where nodes exchange enough information with each other to extract consistent data and then use that consistent data in their main task. Since in our model, there is no communication link between workers, they cannot detect the adversarial data owners and then mitigate their effect by exchanging information with one another. Therefore, the first approach alone is not sufficient.
In this work, we choose a combination of these approaches for our particular model: we deal with inconsistencies in the final step in decoding, with help of tags that workers had previously produced and communicated to the data owners. 

The structure of this paper is as follows. We formulate the problem in Section \ref{problem formulation}, and state the main result in Section \ref{main result section}. Then, explain the concept of tag functions in Section \ref{tag functions section} and introduce a good tag function. We define some notions and analyze the system in Section \ref{definitions section}. Finally, we state the fundamental limit of the system in Section \ref{main result section}.

\section{problem formulation}\label{problem formulation}
In this section, we formally introduce $\textrm{Vers}(f,N,K,\beta,v,\textrm{Enc})$, a fully distributed linear coded computing system with parameters $f,N,K,\beta,v$, and encoding algorithm Enc. Consider a system of $K\in\mathbb{N}$ data owners and $N\in\mathbb{N}$ worker nodes, indexed $1,\dots,K$, and $1,\dots,N$, where $K\le N$. We denote the set $\{1,\dots,m\}, m\in\mathbb{N}$ with $[m]$, and $\{0,1,\dots,m\}$ with $[m^+]$.
The data owner $k\in[K]$ has $X_k\in\mathbb{U}$ and wants to have $f(X_k)$, where $f:\mathbb{U}\rightarrow\mathbb{V}$ is an arbitrary polynomial function of degree $d\in\mathbb{N}$, called \textit{target function}, and $\mathbb{U}, \mathbb{V}$ are vector spaces over a finite field $\mathbb{F}$. Data owners employ workers to do the calculations for them, in a distributed manner.
Among the data owners, there are at most $\beta\in\mathbb{N}$ adversarial nodes, $\beta<K$. We denote the set of adversarial data owners by $\mathcal{A}\in[K]$, and the set of honest data owners by $\mathcal{H}=[K]\setminus\mathcal{A}$. The adversarial data owners try to cause discrepancies in order to mislead honest data owners about the correct values of $f(X_k)$, $k\in\mathcal{H}$. The adversarial data owners sacrifice their chance to use workers for their calculations, and hence, do not want any specific correct computation at the end. The adversarial data owners are free to cooperate with each other, but they do not know the data of the honest nodes. We assume that all workers are honest. Moreover, honest data owners and workers do not know the adversarial data owners. 

We design Vers by using inspirations from \cite{yu2019lagrange}. In the following, we describe the workflow of Vers in detail. The workflow is also shown in Algorithm \ref{Vers algo}. In the first step, data owners transmit their data to the workers. Let $X_{k,n}$ be the message sent from data owner $k\in[K]$ to worker node $n\in[N]$. An honest data owner $k\in\mathcal{H}$ sends $X_k$ to all workers, i.e. $X_{k,n}=X_k$ for all $n\in[N]$. We assume that $X_k, k\in\mathcal{H}$ are chosen independently and uniformly at random from $\mathbb{U}$. An adversarial data owner $k\in\mathcal{A}$ generates at most $v\in\mathbb{N}$ different messages, which we denote by \textit{versions}, $X_k^{(1)},\dots,X_k^{(v)}$, and sends one of them to each worker. In other words, $\{X_{k,n}, n\in[N]\}=\{X_k^{(1)},\dots,X_k^{(v)}\}$, for $k\in\mathcal{A}$. Adversarial data owners can generate up to $v$ messages of their choice and can send any one of them to any worker they choose. The adversarial data owners may also collude so that their aggregate adversarial behavior is as detrimental as possible.

In the second step, worker node $n$ who has received $X_{1,n},\dots,X_{K,n}$ from data owners, computes a linear combination of the $K$ received messages, denoted by $W_n\in\mathbb{U}$, using the encoding algorithm Enc. Let $\gamma_{k,n}\in\mathbb{F}, k\in[K]$ be the $K$ encoding coefficients of worker $n$ that the encoding algorithm Enc determines. Therefore,
    \begin{align}\label{W formula}
        W_n = \sum_{k\in[K]}\gamma_{k,n}X_{k,n},\quad n\in[N].
    \end{align}
We emphasize that encoding in Vers in done in a decentralized fashion, and locally in each worker. This is the main idea of Vers that differentiates it from other coded computing systems that incorporate a central encoder. After local encoding of the received messages, in step \ref{step 3}, worker $n$ computes $f(W_n)$, where $f$ is the target function. 

In addition to $W_n$, and $f(W_n)$, worker node $n$ computes $\mathsf{tag}_n=J(X_{1,n},\dots,X_{K,n})$ in step \ref{step 4}, where $J:\mathbb{U}\rightarrow \mathbb{W}$ is a tag function, and $\mathbb{W}$ is a vector space over $\mathbb{F}$. We will formally introduce tag function in Section \ref{tag functions section}. Intuitively, a tag function compresses its $K$ inputs, and outputs a \textit{fingerprint} of inputs, so that outputs of two different sets of $K$ inputs are almost always different.

In the fifth step, worker node $n$ sends $\mathsf{tag}_n$ and $f(W_n)$ to all data owners. Honest data owners can use tags to identify workers whose received messages were the same in the first step, and therefore, their returned computation results are consistent. In the sixth step, the honest data owner $k\in[K]$ uses $\{f(W_n), n\in\mathcal{T}\}\cup \{\mathsf{tag}_n,n\in\mathcal{T}\}$, where $\mathcal{T}$ is an arbitrary subset of $[N]$ of size $t\in\mathbb{N}$, to recover $f(X_k)$. An example  $\textrm{Vers}(f,N=5,K=3,\beta=1,v=2,\textrm{Enc})$ is shown in Fig \ref{fully decentral fig}. 

\begin{algorithm}
\caption{The workflow of Vers}\label{Vers algo}
\begin{algorithmic}[1]
    \item  The honest data owner $k\in\mathcal{H}$ sends $X_k$ to all workers. The adversarial data owner $k\in\mathcal{A}$ generates at most $v\in\mathbb{N}$ different messages $X_k^{(1)},\dots,X_k^{(v)}$, and sends one of them to each worker. \label{step 1}
    \State 
    Worker node $n\in[N]$ computes $W_n\in\mathbb{U}$ as a linear combination of the $K$ received messages from data owners, using the encoding algorithm Enc, and according to \eqref{W formula}. \label{step 2}
    \State 
    Worker node $n\in[N]$ computes $f(W_n)$. \label{step 3}
    \State Worker node $n\in[N]$ computes $\mathsf{tag}_n=J(X_{1,n},\dots,X_{K,n})$, where $J:\mathbb{U}\rightarrow \mathbb{W}$ is a tag function, and $\mathbb{W}$ is a vector space over $\mathbb{F}$. \label{step 4}
    \State 
    Worker node $n\in[N]$ sends $\mathsf{tag}_n$ and $f(W_n)$ to all data owners. \label{step 5}
    \State 
    An honest data owner $k$ recovers $f(X_k)$ from
    $\{f(W_n), n\in\mathcal{T}\}\cup \{\mathsf{tag}_n,n\in\mathcal{T}\}$, where $\mathcal{T}$ is an arbitrary subset of $[N]$ of size $t\in\mathbb{N}$. \label{step 6}
\end{algorithmic}
\end{algorithm}

\definecolor{ududff}{rgb}{0.30196078431372547,0.30196078431372547,1}
\definecolor{ffzzcc}{rgb}{1,0.6,0.8}

\begin{figure}[h!]
\centering
\begin{subfigure}[t]{0.5\linewidth}
\begin{tikzpicture}[x=0.75pt,y=0.75pt,yscale=-1,xscale=1,scale=0.7]
%uncomment if require: \path (0,1067); %set diagram left start at 0, and has height of 1067

%Shape: Circle [id:dp22148036658815884] 
\draw  [fill={rgb, 255:red, 248; green, 81; blue, 81 }  ,fill opacity=1 ] (225,476.64) .. controls (225,462.83) and (236.19,451.64) .. (250,451.64) .. controls (263.81,451.64) and (275,462.83) .. (275,476.64) .. controls (275,490.45) and (263.81,501.64) .. (250,501.64) .. controls (236.19,501.64) and (225,490.45) .. (225,476.64) -- cycle ;
%Shape: Circle [id:dp3550911979617237] 
\draw  [fill={rgb, 255:red, 248; green, 81; blue, 81 }  ,fill opacity=1 ] (326,476.64) .. controls (326,462.83) and (337.19,451.64) .. (351,451.64) .. controls (364.81,451.64) and (376,462.83) .. (376,476.64) .. controls (376,490.45) and (364.81,501.64) .. (351,501.64) .. controls (337.19,501.64) and (326,490.45) .. (326,476.64) -- cycle ;
%Shape: Circle [id:dp023163510568725787] 
\draw  [fill={rgb, 255:red, 248; green, 81; blue, 81 }  ,fill opacity=1 ] (125,475.64) .. controls (125,461.83) and (136.19,450.64) .. (150,450.64) .. controls (163.81,450.64) and (175,461.83) .. (175,475.64) .. controls (175,489.45) and (163.81,500.64) .. (150,500.64) .. controls (136.19,500.64) and (125,489.45) .. (125,475.64) -- cycle ;
%Shape: Circle [id:dp5699284154835318] 
\draw  [fill={rgb, 255:red, 248; green, 81; blue, 81 }  ,fill opacity=1 ] (526,475.64) .. controls (526,461.83) and (537.19,450.64) .. (551,450.64) .. controls (564.81,450.64) and (576,461.83) .. (576,475.64) .. controls (576,489.45) and (564.81,500.64) .. (551,500.64) .. controls (537.19,500.64) and (526,489.45) .. (526,475.64) -- cycle ;
%Shape: Circle [id:dp9977032912035113] 
\draw  [fill={rgb, 255:red, 248; green, 81; blue, 81 }  ,fill opacity=1 ] (426,476.64) .. controls (426,462.83) and (437.19,451.64) .. (451,451.64) .. controls (464.81,451.64) and (476,462.83) .. (476,476.64) .. controls (476,490.45) and (464.81,501.64) .. (451,501.64) .. controls (437.19,501.64) and (426,490.45) .. (426,476.64) -- cycle ;
%Shape: Circle [id:dp2536041331398893] 
\draw  [fill={rgb, 255:red, 245; green, 166; blue, 35 }  ,fill opacity=1 ] (225,325.64) .. controls (225,311.83) and (236.19,300.64) .. (250,300.64) .. controls (263.81,300.64) and (275,311.83) .. (275,325.64) .. controls (275,339.45) and (263.81,350.64) .. (250,350.64) .. controls (236.19,350.64) and (225,339.45) .. (225,325.64) -- cycle ;
%Shape: Circle [id:dp39066623299294423] 
\draw  [fill={rgb, 255:red, 245; green, 166; blue, 35 }  ,fill opacity=1 ] (325,325.64) .. controls (325,311.83) and (336.19,300.64) .. (350,300.64) .. controls (363.81,300.64) and (375,311.83) .. (375,325.64) .. controls (375,339.45) and (363.81,350.64) .. (350,350.64) .. controls (336.19,350.64) and (325,339.45) .. (325,325.64) -- cycle ;
%Shape: Circle [id:dp6911664151669725] 
\draw  [fill={rgb, 255:red, 245; green, 166; blue, 35 }  ,fill opacity=1 ] (424,325.64) .. controls (424,311.83) and (435.19,300.64) .. (449,300.64) .. controls (462.81,300.64) and (474,311.83) .. (474,325.64) .. controls (474,339.45) and (462.81,350.64) .. (449,350.64) .. controls (435.19,350.64) and (424,339.45) .. (424,325.64) -- cycle ;
%Shape: Can [id:dp9277833856027169] 
\draw  [fill={rgb, 255:red, 184; green, 233; blue, 134 }  ,fill opacity=1 ] (250,325.64) -- (250,338.14) .. controls (250,339.38) and (246.64,340.39) .. (242.5,340.39) .. controls (238.36,340.39) and (235,339.38) .. (235,338.14) -- (235,325.64) .. controls (235,324.4) and (238.36,323.39) .. (242.5,323.39) .. controls (246.64,323.39) and (250,324.4) .. (250,325.64) .. controls (250,326.88) and (246.64,327.89) .. (242.5,327.89) .. controls (238.36,327.89) and (235,326.88) .. (235,325.64) ;
%Shape: Can [id:dp2703556698027032] 
\draw  [fill={rgb, 255:red, 184; green, 233; blue, 134 }  ,fill opacity=1 ] (350,325.64) -- (350,338.14) .. controls (350,339.38) and (346.64,340.39) .. (342.5,340.39) .. controls (338.36,340.39) and (335,339.38) .. (335,338.14) -- (335,325.64) .. controls (335,324.4) and (338.36,323.39) .. (342.5,323.39) .. controls (346.64,323.39) and (350,324.4) .. (350,325.64) .. controls (350,326.88) and (346.64,327.89) .. (342.5,327.89) .. controls (338.36,327.89) and (335,326.88) .. (335,325.64) ;
%Shape: Can [id:dp5312925712939887] 
\draw  [fill={rgb, 255:red, 184; green, 233; blue, 134 }  ,fill opacity=1 ] (449,325.64) -- (449,338.14) .. controls (449,339.38) and (445.64,340.39) .. (441.5,340.39) .. controls (437.36,340.39) and (434,339.38) .. (434,338.14) -- (434,325.64) .. controls (434,324.4) and (437.36,323.39) .. (441.5,323.39) .. controls (445.64,323.39) and (449,324.4) .. (449,325.64) .. controls (449,326.88) and (445.64,327.89) .. (441.5,327.89) .. controls (437.36,327.89) and (434,326.88) .. (434,325.64) ;
%Straight Lines [id:da2791824541076813] 
\draw    (237,360.64) -- (174.19,445.04) ;
\draw [shift={(173,446.64)}, rotate = 306.65999999999997] [color={rgb, 255:red, 0; green, 0; blue, 0 }  ][line width=0.75]    (10.93,-3.29) .. controls (6.95,-1.4) and (3.31,-0.3) .. (0,0) .. controls (3.31,0.3) and (6.95,1.4) .. (10.93,3.29)   ;
%Straight Lines [id:da17683292869350242] 
\draw    (248,360.64) -- (247.02,439.64) ;
\draw [shift={(247,441.64)}, rotate = 270.71] [color={rgb, 255:red, 0; green, 0; blue, 0 }  ][line width=0.75]    (10.93,-3.29) .. controls (6.95,-1.4) and (3.31,-0.3) .. (0,0) .. controls (3.31,0.3) and (6.95,1.4) .. (10.93,3.29)   ;
%Straight Lines [id:da6143771029609386] 
\draw    (260,360.64) -- (330.73,447.09) ;
\draw [shift={(332,448.64)}, rotate = 230.71] [color={rgb, 255:red, 0; green, 0; blue, 0 }  ][line width=0.75]    (10.93,-3.29) .. controls (6.95,-1.4) and (3.31,-0.3) .. (0,0) .. controls (3.31,0.3) and (6.95,1.4) .. (10.93,3.29)   ;
%Straight Lines [id:da24339388982634902] 
\draw    (268,353.64) -- (429.27,447.63) ;
\draw [shift={(431,448.64)}, rotate = 210.23] [color={rgb, 255:red, 0; green, 0; blue, 0 }  ][line width=0.75]    (10.93,-3.29) .. controls (6.95,-1.4) and (3.31,-0.3) .. (0,0) .. controls (3.31,0.3) and (6.95,1.4) .. (10.93,3.29)   ;
%Straight Lines [id:da4690380883884633] 
\draw    (274,348.64) -- (529.12,442.95) ;
\draw [shift={(531,443.64)}, rotate = 200.29] [color={rgb, 255:red, 0; green, 0; blue, 0 }  ][line width=0.75]    (10.93,-3.29) .. controls (6.95,-1.4) and (3.31,-0.3) .. (0,0) .. controls (3.31,0.3) and (6.95,1.4) .. (10.93,3.29)   ;
%Straight Lines [id:da8754695064572371] 
\draw [color={rgb, 255:red, 163; green, 157; blue, 157 }  ,draw opacity=1 ][fill={rgb, 255:red, 245; green, 234; blue, 234 }  ,fill opacity=1 ] [dash pattern={on 4.5pt off 4.5pt}]  (330,350.64) -- (178.66,451.53) ;
\draw [shift={(177,452.64)}, rotate = 326.31] [color={rgb, 255:red, 163; green, 157; blue, 157 }  ,draw opacity=1 ][line width=0.75]    (10.93,-3.29) .. controls (6.95,-1.4) and (3.31,-0.3) .. (0,0) .. controls (3.31,0.3) and (6.95,1.4) .. (10.93,3.29)   ;
%Straight Lines [id:da9289964335894505] 
\draw [color={rgb, 255:red, 163; green, 157; blue, 157 }  ,draw opacity=1 ][fill={rgb, 255:red, 245; green, 234; blue, 234 }  ,fill opacity=1 ] [dash pattern={on 4.5pt off 4.5pt}]  (340,355.64) -- (264.32,442.14) ;
\draw [shift={(263,443.64)}, rotate = 311.19] [color={rgb, 255:red, 163; green, 157; blue, 157 }  ,draw opacity=1 ][line width=0.75]    (10.93,-3.29) .. controls (6.95,-1.4) and (3.31,-0.3) .. (0,0) .. controls (3.31,0.3) and (6.95,1.4) .. (10.93,3.29)   ;
%Straight Lines [id:da014374072429092388] 
\draw [color={rgb, 255:red, 163; green, 157; blue, 157 }  ,draw opacity=1 ][fill={rgb, 255:red, 245; green, 234; blue, 234 }  ,fill opacity=1 ] [dash pattern={on 4.5pt off 4.5pt}]  (351,357.64) -- (351.98,442.64) ;
\draw [shift={(352,444.64)}, rotate = 269.34000000000003] [color={rgb, 255:red, 163; green, 157; blue, 157 }  ,draw opacity=1 ][line width=0.75]    (10.93,-3.29) .. controls (6.95,-1.4) and (3.31,-0.3) .. (0,0) .. controls (3.31,0.3) and (6.95,1.4) .. (10.93,3.29)   ;
%Straight Lines [id:da398810962277085] 
\draw [color={rgb, 255:red, 163; green, 157; blue, 157 }  ,draw opacity=1 ][fill={rgb, 255:red, 245; green, 234; blue, 234 }  ,fill opacity=1 ] [dash pattern={on 4.5pt off 4.5pt}]  (364,357.64) -- (438.7,445.12) ;
\draw [shift={(440,446.64)}, rotate = 229.5] [color={rgb, 255:red, 163; green, 157; blue, 157 }  ,draw opacity=1 ][line width=0.75]    (10.93,-3.29) .. controls (6.95,-1.4) and (3.31,-0.3) .. (0,0) .. controls (3.31,0.3) and (6.95,1.4) .. (10.93,3.29)   ;
%Straight Lines [id:da25786862131852706] 
\draw [color={rgb, 255:red, 163; green, 157; blue, 157 }  ,draw opacity=1 ][fill={rgb, 255:red, 245; green, 234; blue, 234 }  ,fill opacity=1 ] [dash pattern={on 0.84pt off 2.51pt}]  (470,352.64) -- (543.7,439.12) ;
\draw [shift={(545,440.64)}, rotate = 229.56] [color={rgb, 255:red, 163; green, 157; blue, 157 }  ,draw opacity=1 ][line width=0.75]    (10.93,-3.29) .. controls (6.95,-1.4) and (3.31,-0.3) .. (0,0) .. controls (3.31,0.3) and (6.95,1.4) .. (10.93,3.29)   ;
%Straight Lines [id:da6563705849705792] 
\draw [color={rgb, 255:red, 163; green, 157; blue, 157 }  ,draw opacity=1 ][fill={rgb, 255:red, 245; green, 234; blue, 234 }  ,fill opacity=1 ] [dash pattern={on 0.84pt off 2.51pt}]  (457,359.64) -- (455.05,443.64) ;
\draw [shift={(455,445.64)}, rotate = 271.33] [color={rgb, 255:red, 163; green, 157; blue, 157 }  ,draw opacity=1 ][line width=0.75]    (10.93,-3.29) .. controls (6.95,-1.4) and (3.31,-0.3) .. (0,0) .. controls (3.31,0.3) and (6.95,1.4) .. (10.93,3.29)   ;
%Straight Lines [id:da8218786677381247] 
\draw [color={rgb, 255:red, 163; green, 157; blue, 157 }  ,draw opacity=1 ][fill={rgb, 255:red, 245; green, 234; blue, 234 }  ,fill opacity=1 ] [dash pattern={on 0.84pt off 2.51pt}]  (433,354.64) -- (368.17,445.02) ;
\draw [shift={(367,446.64)}, rotate = 305.65999999999997] [color={rgb, 255:red, 163; green, 157; blue, 157 }  ,draw opacity=1 ][line width=0.75]    (10.93,-3.29) .. controls (6.95,-1.4) and (3.31,-0.3) .. (0,0) .. controls (3.31,0.3) and (6.95,1.4) .. (10.93,3.29)   ;
%Straight Lines [id:da11336288103363357] 
\draw [color={rgb, 255:red, 163; green, 157; blue, 157 }  ,draw opacity=1 ][fill={rgb, 255:red, 245; green, 234; blue, 234 }  ,fill opacity=1 ] [dash pattern={on 0.84pt off 2.51pt}]  (426,346.64) -- (277.64,450.49) ;
\draw [shift={(276,451.64)}, rotate = 325.01] [color={rgb, 255:red, 163; green, 157; blue, 157 }  ,draw opacity=1 ][line width=0.75]    (10.93,-3.29) .. controls (6.95,-1.4) and (3.31,-0.3) .. (0,0) .. controls (3.31,0.3) and (6.95,1.4) .. (10.93,3.29)   ;
%Straight Lines [id:da8259287807707274] 
\draw [color={rgb, 255:red, 163; green, 157; blue, 157 }  ,draw opacity=1 ][fill={rgb, 255:red, 245; green, 234; blue, 234 }  ,fill opacity=1 ] [dash pattern={on 0.84pt off 2.51pt}]  (423,338.64) -- (184.79,457.75) ;
\draw [shift={(183,458.64)}, rotate = 333.43] [color={rgb, 255:red, 163; green, 157; blue, 157 }  ,draw opacity=1 ][line width=0.75]    (10.93,-3.29) .. controls (6.95,-1.4) and (3.31,-0.3) .. (0,0) .. controls (3.31,0.3) and (6.95,1.4) .. (10.93,3.29)   ;
%Straight Lines [id:da6464117648328491] 
\draw [color={rgb, 255:red, 163; green, 157; blue, 157 }  ,draw opacity=1 ][fill={rgb, 255:red, 245; green, 234; blue, 234 }  ,fill opacity=1 ] [dash pattern={on 4.5pt off 4.5pt}]  (375,350.64) -- (522.37,454.49) ;
\draw [shift={(524,455.64)}, rotate = 215.17000000000002] [color={rgb, 255:red, 163; green, 157; blue, 157 }  ,draw opacity=1 ][line width=0.75]    (10.93,-3.29) .. controls (6.95,-1.4) and (3.31,-0.3) .. (0,0) .. controls (3.31,0.3) and (6.95,1.4) .. (10.93,3.29)   ;

% Text Node
\draw (172,377.64) node [anchor=north west][inner sep=0.75pt]   [align=left] {$\displaystyle X_{1}^{( 1)}$};
% Text Node
\draw (216,394.64) node [anchor=north west][inner sep=0.75pt]   [align=left] {$\displaystyle X_{1}^{( 1)}$};
% Text Node
\draw (265,401.64) node [anchor=north west][inner sep=0.75pt]   [align=left] {$\displaystyle X_{1}^{( 2)}$};
% Text Node
\draw (316,395.64) node [anchor=north west][inner sep=0.75pt]   [align=left] {$\displaystyle X_{1}^{( 2)}$};
% Text Node
\draw (401,372.64) node [anchor=north west][inner sep=0.75pt]   [align=left] {$\displaystyle X_{1}^{( 2)}$};
\end{tikzpicture}
\caption{The orange nodes (top row) are data owners, and the red nodes (bottom row) are workers. Data owners send their data to workers. Here, the first data owner is an adversary and sends two different messages $X_1^{(1)}$ and $X_1^{(2)}$ to workers. }
\label{fullt decentral fig a}
\end{subfigure}    
\hfill
\begin{subfigure}[t]{0.5\linewidth}
\centering
\begin{tikzpicture}[x=0.75pt,y=0.75pt,yscale=-1,xscale=1,scale=0.7]
%uncomment if require: \path (0,1067); %set diagram left start at 0, and has height of 1067

%Shape: Circle [id:dp2525254395279344] 
\draw  [fill={rgb, 255:red, 248; green, 81; blue, 81 }  ,fill opacity=1 ] (225,725.64) .. controls (225,711.83) and (236.19,700.64) .. (250,700.64) .. controls (263.81,700.64) and (275,711.83) .. (275,725.64) .. controls (275,739.45) and (263.81,750.64) .. (250,750.64) .. controls (236.19,750.64) and (225,739.45) .. (225,725.64) -- cycle ;
%Shape: Circle [id:dp3719962612300922] 
\draw  [fill={rgb, 255:red, 248; green, 81; blue, 81 }  ,fill opacity=1 ] (326,725.64) .. controls (326,711.83) and (337.19,700.64) .. (351,700.64) .. controls (364.81,700.64) and (376,711.83) .. (376,725.64) .. controls (376,739.45) and (364.81,750.64) .. (351,750.64) .. controls (337.19,750.64) and (326,739.45) .. (326,725.64) -- cycle ;
%Shape: Circle [id:dp1965166543673289] 
\draw  [fill={rgb, 255:red, 248; green, 81; blue, 81 }  ,fill opacity=1 ] (125,724.64) .. controls (125,710.83) and (136.19,699.64) .. (150,699.64) .. controls (163.81,699.64) and (175,710.83) .. (175,724.64) .. controls (175,738.45) and (163.81,749.64) .. (150,749.64) .. controls (136.19,749.64) and (125,738.45) .. (125,724.64) -- cycle ;
%Shape: Circle [id:dp973043510887865] 
\draw  [fill={rgb, 255:red, 248; green, 81; blue, 81 }  ,fill opacity=1 ] (526,724.64) .. controls (526,710.83) and (537.19,699.64) .. (551,699.64) .. controls (564.81,699.64) and (576,710.83) .. (576,724.64) .. controls (576,738.45) and (564.81,749.64) .. (551,749.64) .. controls (537.19,749.64) and (526,738.45) .. (526,724.64) -- cycle ;
%Shape: Circle [id:dp7370092235420034] 
\draw  [fill={rgb, 255:red, 248; green, 81; blue, 81 }  ,fill opacity=1 ] (426,725.64) .. controls (426,711.83) and (437.19,700.64) .. (451,700.64) .. controls (464.81,700.64) and (476,711.83) .. (476,725.64) .. controls (476,739.45) and (464.81,750.64) .. (451,750.64) .. controls (437.19,750.64) and (426,739.45) .. (426,725.64) -- cycle ;
%Shape: Circle [id:dp6044338788084114] 
\draw  [fill={rgb, 255:red, 245; green, 166; blue, 35 }  ,fill opacity=1 ] (225,574.64) .. controls (225,560.83) and (236.19,549.64) .. (250,549.64) .. controls (263.81,549.64) and (275,560.83) .. (275,574.64) .. controls (275,588.45) and (263.81,599.64) .. (250,599.64) .. controls (236.19,599.64) and (225,588.45) .. (225,574.64) -- cycle ;
%Shape: Circle [id:dp3753933040595352] 
\draw  [fill={rgb, 255:red, 245; green, 166; blue, 35 }  ,fill opacity=1 ] (325,574.64) .. controls (325,560.83) and (336.19,549.64) .. (350,549.64) .. controls (363.81,549.64) and (375,560.83) .. (375,574.64) .. controls (375,588.45) and (363.81,599.64) .. (350,599.64) .. controls (336.19,599.64) and (325,588.45) .. (325,574.64) -- cycle ;
%Shape: Circle [id:dp5330528555563248] 
\draw  [fill={rgb, 255:red, 245; green, 166; blue, 35 }  ,fill opacity=1 ] (424,574.64) .. controls (424,560.83) and (435.19,549.64) .. (449,549.64) .. controls (462.81,549.64) and (474,560.83) .. (474,574.64) .. controls (474,588.45) and (462.81,599.64) .. (449,599.64) .. controls (435.19,599.64) and (424,588.45) .. (424,574.64) -- cycle ;
%Shape: Can [id:dp013525357051446862] 
\draw  [fill={rgb, 255:red, 184; green, 233; blue, 134 }  ,fill opacity=1 ] (250,574.64) -- (250,587.14) .. controls (250,588.38) and (246.64,589.39) .. (242.5,589.39) .. controls (238.36,589.39) and (235,588.38) .. (235,587.14) -- (235,574.64) .. controls (235,573.4) and (238.36,572.39) .. (242.5,572.39) .. controls (246.64,572.39) and (250,573.4) .. (250,574.64) .. controls (250,575.88) and (246.64,576.89) .. (242.5,576.89) .. controls (238.36,576.89) and (235,575.88) .. (235,574.64) ;
%Shape: Can [id:dp48857354655616336] 
\draw  [fill={rgb, 255:red, 184; green, 233; blue, 134 }  ,fill opacity=1 ] (350,574.64) -- (350,587.14) .. controls (350,588.38) and (346.64,589.39) .. (342.5,589.39) .. controls (338.36,589.39) and (335,588.38) .. (335,587.14) -- (335,574.64) .. controls (335,573.4) and (338.36,572.39) .. (342.5,572.39) .. controls (346.64,572.39) and (350,573.4) .. (350,574.64) .. controls (350,575.88) and (346.64,576.89) .. (342.5,576.89) .. controls (338.36,576.89) and (335,575.88) .. (335,574.64) ;
%Shape: Can [id:dp04764252265012159] 
\draw  [fill={rgb, 255:red, 184; green, 233; blue, 134 }  ,fill opacity=1 ] (449,574.64) -- (449,587.14) .. controls (449,588.38) and (445.64,589.39) .. (441.5,589.39) .. controls (437.36,589.39) and (434,588.38) .. (434,587.14) -- (434,574.64) .. controls (434,573.4) and (437.36,572.39) .. (441.5,572.39) .. controls (445.64,572.39) and (449,573.4) .. (449,574.64) .. controls (449,575.88) and (445.64,576.89) .. (441.5,576.89) .. controls (437.36,576.89) and (434,575.88) .. (434,574.64) ;
%Straight Lines [id:da2549464927969891] 
\draw    (168,693.56) -- (226.88,606.22) ;
\draw [shift={(228,604.56)}, rotate = 483.99] [color={rgb, 255:red, 0; green, 0; blue, 0 }  ][line width=0.75]    (10.93,-3.29) .. controls (6.95,-1.4) and (3.31,-0.3) .. (0,0) .. controls (3.31,0.3) and (6.95,1.4) .. (10.93,3.29)   ;
%Straight Lines [id:da2979272546820342] 
\draw [color={rgb, 255:red, 126; green, 124; blue, 124 }  ,draw opacity=1 ][fill={rgb, 255:red, 245; green, 234; blue, 234 }  ,fill opacity=1 ] [dash pattern={on 4.5pt off 4.5pt}]  (335,692.56) -- (255.39,610) ;
\draw [shift={(254,608.56)}, rotate = 406.03999999999996] [color={rgb, 255:red, 126; green, 124; blue, 124 }  ,draw opacity=1 ][line width=0.75]    (10.93,-3.29) .. controls (6.95,-1.4) and (3.31,-0.3) .. (0,0) .. controls (3.31,0.3) and (6.95,1.4) .. (10.93,3.29)   ;
%Straight Lines [id:da25012561027336333] 
\draw    (177,699.56) -- (328.29,607.6) ;
\draw [shift={(330,606.56)}, rotate = 508.71] [color={rgb, 255:red, 0; green, 0; blue, 0 }  ][line width=0.75]    (10.93,-3.29) .. controls (6.95,-1.4) and (3.31,-0.3) .. (0,0) .. controls (3.31,0.3) and (6.95,1.4) .. (10.93,3.29)   ;
%Straight Lines [id:da649431513990979] 
\draw    (181,709.56) -- (428.16,605.34) ;
\draw [shift={(430,604.56)}, rotate = 517.14] [color={rgb, 255:red, 0; green, 0; blue, 0 }  ][line width=0.75]    (10.93,-3.29) .. controls (6.95,-1.4) and (3.31,-0.3) .. (0,0) .. controls (3.31,0.3) and (6.95,1.4) .. (10.93,3.29)   ;
%Straight Lines [id:da9469358758790674] 
\draw [color={rgb, 255:red, 126; green, 124; blue, 124 }  ,draw opacity=1 ][fill={rgb, 255:red, 245; green, 234; blue, 234 }  ,fill opacity=1 ] [dash pattern={on 4.5pt off 4.5pt}]  (357,692.56) -- (358.95,612.56) ;
\draw [shift={(359,610.56)}, rotate = 451.4] [color={rgb, 255:red, 126; green, 124; blue, 124 }  ,draw opacity=1 ][line width=0.75]    (10.93,-3.29) .. controls (6.95,-1.4) and (3.31,-0.3) .. (0,0) .. controls (3.31,0.3) and (6.95,1.4) .. (10.93,3.29)   ;
%Straight Lines [id:da4742047669382219] 
\draw [color={rgb, 255:red, 126; green, 124; blue, 124 }  ,draw opacity=1 ][fill={rgb, 255:red, 245; green, 234; blue, 234 }  ,fill opacity=1 ] [dash pattern={on 4.5pt off 4.5pt}]  (373,699.56) -- (438.8,612.16) ;
\draw [shift={(440,610.56)}, rotate = 486.97] [color={rgb, 255:red, 126; green, 124; blue, 124 }  ,draw opacity=1 ][line width=0.75]    (10.93,-3.29) .. controls (6.95,-1.4) and (3.31,-0.3) .. (0,0) .. controls (3.31,0.3) and (6.95,1.4) .. (10.93,3.29)   ;

% Text Node
\draw (140,619.64) node [anchor=north west][inner sep=0.75pt]   [align=left] {$\displaystyle  \begin{array}{{>{\displaystyle}l}}
f( W_{1})\\
\text{tag}_{1}
\end{array}$};
% Text Node
\draw (411,642.64) node [anchor=north west][inner sep=0.75pt]   [align=left] {$\displaystyle  \begin{array}{{>{\displaystyle}l}}
\textcolor[rgb]{0.45,0.42,0.42}{f( W_{3})}\\
\textcolor[rgb]{0.45,0.42,0.42}{\text{tag}_{3}}
\end{array}$};
\end{tikzpicture}
\caption{Workers apply $f$ on coded input, calculate tags, and return the results to data owners.}
\label{fully decentral fig b}
\end{subfigure}
\caption{$\textrm{Vers}(f,N=5,K=3,\beta=1,v=2,\textrm{Enc})$ is shown as an example (only some of the messages between data owners and workers are shown). The leftmost data owner is adversarial.}
\label{fully decentral fig}
\end{figure}
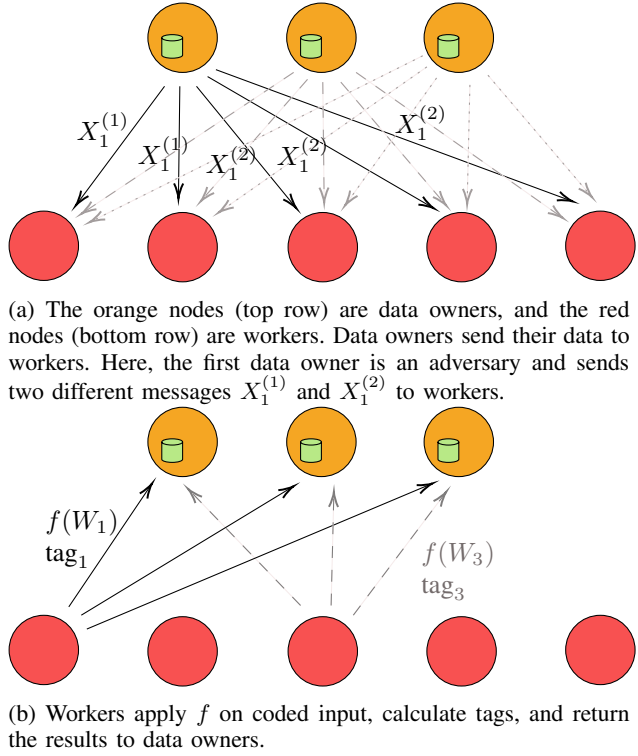

\begin{remark}
Note that each worker node calculates a single tag in the fourth step of Vers. Therefore, the total communication of Vers is $O(N)$, i.e. linear in $N$. Data owners use the indirect information from these tags to partition workers into sets that have received the same adversarial data from the adversarial data owners. Another possible design is having each worker calculate $K$ tags, one for each data received from the $K$ data owners. The direct information from such tags reveals the identities of adversarial data owners for honest data owners. However, such design results in $O(KN)$ communication load, and hinders scalability when $K$ grows. As previously stated in Section \ref{introduction section}, one of the main purposes of Vers is to compensate for scalability bottlenecks of single master systems. Therefore, in Vers, we incorporate the communication efficient single tag method, to allow $K$ to scale.
Another intermediate approach is to have each worker calculate a constant number of tags of the $K$ received messages. This method probably gives more information about the adversarial behaviour to honest data owners. However, due to its complexity, we do not address this approach in this paper, and leave it for future research.
\end{remark}
The formal definition of fundamental limit is as following.
\begin{definition}[The fundamental limit of Vers]\label{The fundamental limit of Vers}
The fundamental limit of $\textrm{Vers}(f,N,K,\beta,v,\textrm{Enc})$, which we denote by $t^*(f,N,K,\beta,v,\textrm{Enc})$, is the minimum $t$ required such that $f(X_k), k\in\mathcal{H}$ can be correctly computed from $\{f(W_n), n\in\mathcal{T}\}$, and $\{\mathsf{tag}_n, n\in\mathcal{T}\}$, under any adversarial behavior, where $\mathcal{T}$ is any arbitrary subset of $[N]$ of size $t$.
\end{definition}
No matter how the adversaries choose and distribute their $v$ different messages $X_k^{(1)},\dots,X_k^{(v)}, k\in\mathcal{A}$ to workers, $f(X_k), k\in\mathcal{H}$ should be correctly extractable from $\{f(W_n), n\in\mathcal{T}\}$ and $\{\mathsf{tag}_n, n\in\mathcal{T}\}$. This should hold for any subset of workers of size $t^*(f,N,K,\beta,v,\textrm{Enc})$, for the system to be robust against stragglers and node failures which is common in distributed systems. Note that we do not set any requirement on being able to find all or any of $f(X_k^{(1)}),\dots,f(X_k^{(v)})$, $k\in\mathcal{A}$, from $\{f(W_n), n\in\mathcal{T}\}$, and $\{\mathsf{tag}_n, n\in\mathcal{T}\}$. This is because the adversaries give up their chance of using workers for computation, to cause inconsistency in the system. In fact, there exist adversarial behaviors for which it is not possible to find all of $f(X_k^{(1)}),\dots,f(X_k^{(v)})$. For example, consider the case where all adversarial data owners $k\in\mathcal{A}$ send $X_k^{(1)}$ to workers $1,\dots,N-1$, and $X_k^{(2)}$ to worker $N$. In this case, there is not enough information to find $f(X_k^{(2)}), k\in\mathcal{A}$ from any set of equations. Therefore, in order to avoid complex scenarios, we define the fundamental limit based on the requirement to find only $f(X_k), k\in\mathcal{H}$. In Section \ref{main result section}, we characterize the fundamental limit of $\textrm{Vers}(f,N,K,\beta,v,\textrm{Lagrange})$. In other words, we choose Lagrange encoding as Enc algorithm in $\textrm{Vers}(f,N,K,\beta,v,\textrm{Enc})$, set $\gamma_{k,n}$ in \eqref{W formula} accordingly, and then find $t^*(f,N,K,\beta,v,\textrm{Lagrange})$.

\begin{remark}
The challenging part of Vers is the last step where the data owners decode the results of workers because the adversarial data owners had injected inconsistencies in the received data of workers. In this work, data owners deal with inconsistencies in decoding and use tags to determine subsets of workers whose results are consistent. Another approach is to resolve the inconsistencies right in the first step, by using a \textit{reliable broadcast} algorithm, e.g. \cite{bracha1987asynchronous}, for the message of each data owner, instead of having data owners simply send messages to workers. This approach requires communication between workers to help them reach a consensus on the received data from each data owner. 
\end{remark}

\section{Applications of Vers}
The motivation for studying Vers is that we can use it to model a variety of emerging computing systems in the presence of adversaries, such as internet of things networks and blockchains. In this section, we elaborate on some applications of Vers.

In an IoT network, since sensors and devices are resource-constrained, they cannot process the data generated in them. Therefore, they need to offload the computations to external nodes, i.e. workers in Vers, to calculate complex target functions of their data. Thus, Vers can easily be used to model an IoT network of resource-constrained devices. Data owners in Vers are equivalent to devices in IoT network. Workers in Vers are equivalent to external computational nodes in IoT network. The target function in Vers is equivalent to any processing function that devices in IoT network require. For example, suppose that in agriculture, IoT is used for smart irrigation. Some devices need to know the efficient amount of water needed, which is a function of soil moisture, weather conditions and ambient temperature, crop type, and possibly more features. Such function can be modeled by target function in Vers.

The adversarial data owners in Vers are equivalent to adversarial IoT devices, which include tampered and infected ones (e.g. see \cite{cloudflare.com_2023}). The adversarial devices may deliberately send inconsistent data to computational nodes, in order to mislead the decision making process of other devices. For example, suppose that IoT is used to deploy smart home, and some devices have been tampered with and try to affect the correct functionality of the smart lock on the door. Since IoT devices are resource-constrained, the upper limit on the different messages of adversarial data owners in Vers applies to them.
The fundamental limit of the Vers model of an IoT network specifies the number of external computational nodes that are required for the system to be robust against the adversarial devices.

Vers can also be used to model blockchains, and in particular, sharded blockchains. In a blockchain network of $N\in\mathbb{N}$ nodes, nodes collectively process transactions, produce blocks of transactions, distribute the blocks, and also validate them. Blockchain is indeed the agreed-upon chain of blocks, and whenever nodes reach consensus on a new block, that block is appended to the chain. In a fully replicated blockchain system, such as Bitcoin, each node repeats what other nodes do. Such decentralized approach provides high security against adversaries, because no node relies on another node, but its throughput (the number of confirmed transactions per second) fails to scale due to its replication. Over the past few years, there has been extensive research into how to make blockchains scalable  \cite{sanka2021systematic,nasir2022scalable}. A notable line of research in this regard, is \textit{sharding} which is inspired by the idea of parallel computing. Recently, Ethereum, the most popular blockchain platform, introduced \textit{Danksharding} \cite{ethereum.org_2023}, which is a novel sharding protocol, and is going to deploy it in near future.  

In vanilla sharding, e.g. \cite{kokoris2018omniledger,zamani2018rapidchain}, blockchain nodes are divided into groups called shards, say $M\in\mathbb{N}$ shards. shards run in parallel, such that each shard of $\frac{N}{M}$ nodes has a local chain and produces its own blocks. The allocation of nodes to shards can be random, or other methods. The purpose of sharding is to increase the throughput by $M$ times. However, a shard of $\frac{N}{M}$ nodes, and thus the whole sharded system, is more vulnerable to security attacks in comparison to $N$ nodes. For example, carrying out a 51\% attack on a shard of $\frac{N}{M}$ nodes is easier for an adversary, compared to a system of $N$ nodes.
In order to preserve security and scale the blockchain simultaneously, the concept of coded sharding was first introduced in \cite{li2020polyshard} and named \textit{PolyShard}.

Unlike vanilla sharding where nodes in a shard work on their own blocks, in Ployshard, all nodes in all shards work on coded blocks. Polyshard works as following. In round $t\in\mathbb{N}$, each shard $m\in[M]$ needs to validate its new block $B_m^{t}$. Shards broadcast their blocks to all nodes. Each node $n$ calculates a coded block $\Tilde{B}^{t}_n$ from the received blocks $B_1^{t},\dots,B_M^{t}$, using Lagrange coding. To put it more clearly, each node $n$ makes a Lagrange ploynomial whose coefficients are constructed from $B_1^{t},\dots,B_M^{t}$, and then evaluate that in a point $\alpha_n$ to obtain $\Tilde{B}^{t}_n$.
Then, nodes apply the block verification function\footnote{Block verification function is determined by the consensus algorithm of the blockchain system.} $g$ on coded blocks that they calculated previously, instead of applying that $M$ times on all $M$ blocks separately. Validation of blocks in round $t$ requires data of blocks in previous rounds $1,2,\dots,t-1$. Each node $n$ stores coded chain of blocks in each round, evaluated at point $\alpha_n$, rather than only blocks of its own shard. Therefore, it can use its stored coded local chain until round $t-1$ when verifying its coded block in round $t$. Nodes broadcast their results of verification in the network. In this step, adversaries may broadcast arbitrary erroneous values. Nodes can obtain the verification of $B_1^{t},\dots,B_M^{t}$ by decoding the verifications of coded blocks, and the adversaries' effect translates into errors in decoding data. The main result of \cite{li2020polyshard} is that the throughput and security\footnote{The number of adversarial nodes that the system is resilient against.} of this scheme scales linearly with $N$, while the storage requirement of nodes does not need to scale.

The adversarial model in \cite{li2020polyshard} is limited to adversaries that broadcast an incorrect value instead of the verification result of their coded block. Even in that case, the adversaries send the same incorrect value to all nodes, and the consistency of the system is maintained. However, in a decentralized system like blockchain, adversarial nodes can freely deviate from the protocol. In \cite{khooshemehr2021discrepancy}, we introduced the discrepancy attack on PolyShard, in which the adversaries take control of some shards and adversarial shards disseminate different blocks to different nodes. This causes error in coded blocks of nodes and disrupts the validation process, and breaks the scalability of security in Polyshard.

Vers is a comprehensive system that can model a coded sharding blockchain in the presence of adversaries as following. Data owner $k\in[K]$ in Vers is equivalent to the representative of shard $k$ in blockchain that propagates the new block of its shard in the network. Note that under the discrepancy attack, the representative of shard $k$ may be adversarial. Data $X_k$ of data owner $k$ in Vers is equivalent to block $B_k$ of shard $k$ in the current round. Workers in Vers are equivalent to nodes of the network in blockchain. The target function $f$ in Vers is equivalent to block verification function $g$ in blockchain. In Vers, data owner $k$ needs $f(X_k)$, and in blockchain, nodes of shard $k$ need the result of the verification function of $B_k$.
The fundamental limit of the Vers model of a coded sharding blockchain specifies the number of nodes that are required for the system to be robust against the adversarial shards.

\section{Main result}\label{main result section}
The following theorem states the main result.
\begin{theorem}\label{main result}
The fundamental limit of $\textrm{Vers}(f,N,K,\beta,v,\textrm{Lagrange})$ 
is $t^*=v^{\beta}d(K-1)+1$.
\end{theorem}
This result implies that $N\ge v^{\beta}d(K-1)+1$ should hold. Otherwise, even when there is no straggler or faulty worker, and all workers send their computations back to data owners, data owners would not be able to recover $f(X_k), k\in\mathcal{H}$.

Since the fundamental limit is the minimum number of workers whose results are enough to correctly recover $f(X_k), k\in\mathcal{H}$, two proofs are needed for Theorem \ref{main result}. Firstly, in converse proof, we need to prove that $t^*$ is truly the minimum, which means there should exist a particular adversarial behaviour and a set of $t^*-1$ workers whose results do not determine $f(X_k), k\in\mathcal{H}$ uniquely. But this is not as simple as showing a system of equations is underdetermined. Note that the honest data owners do not know the adversarial behaviour, so they do not know the correct way to form equations from the results of workers, even with the help of tags. Consequently, honest data owners need to consider every possible system of equations from the results of workers. Moreover, since Lagrange coding is used, $f(X_k), k\in\mathcal{H}$ are not directly the unknowns of the equations of workers, rather, they are linear combinations of the unknowns. In the converse proof, we show that there exists a possible interpretation of the results of workers that leads to at least one incorrect output for $f(X_k), k\in\mathcal{H}$. We provide the converse proof in Subsection \ref{converse proof}.

Secondly, in the achievable proof, we need to prove that any set of $t^*$ workers is enough to correctly recover $f(X_k), k\in\mathcal{H}$. Again, this should hold for any adversarial behavior, despite the existence of many possible interpretations of the results of workers. We provide the achievable proof in Subsection \ref{achievable proof}.

\begin{remark}\label{baseline remark}
The fundamental limit of $\textrm{Vers}(f,N,K,\beta=0,v,\textrm{Enc})$, i.e. when all data owners are honest, is not different from the fundamental limit of a similar system in which a central entity encodes the data and sends the encoded data to workers, because the encoded data in workers would be the same. Indeed, the challenge in Vers is due to the adversarial data owners that inject contradictory data in the system, and the resulting error that propagates through local encodings in workers. The coded computing system studied in \cite{yu2019lagrange} includes a central encoder, but it is essentially similar to $\textrm{Vers}(f,N,K,\beta=0,v,\textrm{Lagrange})$, and its so called \textit{recovery threshold} is equivalent to the fundamental limit of $\textrm{Vers}(f,N,K,\beta=0,v,\textrm{Lagrange})$. Therefore, using the result of \cite{yu2019lagrange}, we know that the fundamental limit of $\textrm{Vers}(f,N,K,\beta=0,v,\textrm{Lagrange})$ is $t^*(f,N,K,\beta=0,v,\textrm{Lagrange})=d(K-1)+1$, where $d$ is the degree of the target polynomial function $f$.
\end{remark}

\section{Tag functions}\label{tag functions section}
In Algorithm \ref{Vers algo}, step \ref{step 2}, we explained that workers compute a \textit{tag} of their received messages. In this section, we introduce the concept of tags, and prove the existence of a tag function that can be deployed in $\textrm{Vers}(f,N,K,\beta,v,\textrm{Enc})$.

Recall that each adversarial data owner $k\in\mathcal{A}$ distributes $X_k^{(1)},\dots,X_k^{(v)}$ among workers. The honest data owners that receive $f(W_n)$ from a worker $n\in[N]$ after step \ref{step 5} of Algorithm \ref{Vers algo}, do not know which one of the $v$ messages of the adversaries this worker has received and has used in $W_n$. Therefore, honest data owners do not know which workers have used the same set of input messages in their calculations. The purpose of tag is to inform the honest data owners about the data used by workers. Tags allow data owners to partition $[N]$ into some sets, at most $v^{\beta}$ sets, where each set contains workers that have received the same messages from all data owners in the first step. 

Let us denote the concatenation of all $X_k, k\in\mathcal{H}$, i.e. the data of the honest data owners, with $X_{\mathcal{H}}$. Also let $X_{\mathcal{A},n}$ be the concatenation of all $X_{k,n}, k\in\mathcal{A}$, i.e. the messages of the adversarial data owners received by worker $n\in[N]$. 
Worker $n\in[N]$ applies the tag function $J$ on $X_{1,n},\dots,X_{K,n}$, but since $\{X_{1,n},\dots,X_{K,n}\}=\{X_{\mathcal{H}},X_{\mathcal{A},n}\}$ for all $n\in[N]$, we exploit the notation and use $J(X_{\mathcal{H}},X_{\mathcal{A},n})$. We define a tag function formally as follows.
\begin{definition}[Tag Function]\label{tag definition}
A function $J:\mathbb{U}^K\rightarrow\mathbb{U}^l$, $l<k$, is a tag function if for two adversaries that choose two different $X_{A,n_1}$ and $X_{A,n_2}$ independently from $X_{\mathcal{H}}$, and send them to workers $n_1,n_2\in[N]$,   
\begin{align}\label{J definition}
\textrm{Pr}\big(J(X_{\mathcal{H}},X_{\mathcal{A},n_1})=J(X_{\mathcal{H}},X_{\mathcal{A},n_2})\big) < \epsilon 
\end{align}
holds, where $\epsilon\in(0,1)$ is a negligible value.
\end{definition}
There are some notable remarks about this definition.
\begin{itemize}
    \item 
    The probability in \eqref{J definition} is over $X_{\mathcal{H}}$, because the data of honest data owners come from an i.i.d uniform distribution on $\mathbb{U}$.
    \item
    As mentioned in section \ref{problem formulation}, adversaries are unaware of the data of the honest data owners. Therefore, they can not choose $X_{A,n_1}$ and $X_{A,n_2}$ based on the knowledge of $X_{\mathcal{H}}$, and they have to choose $X_{A,n_1}$ and $X_{A,n_2}$ independently from $X_{\mathcal{H}}$.
    \item
    The condition $l<k$ is because we want the tag to be lightweight and communication-efficient. If we do not impose such constraint, each worker $n$ could use the concatenation of all $X_{k,n}$, $k\in[K]$ as a tag, and data owners could use such tag very easily to spot the inconsistencies between workers. However, this trivial solution is not communication efficient at all. 
    We need a tag function that compresses the received $K$ messages in each worker so that the discrepancies between workers caused by adversaries become evident for honest data owners.
\end{itemize}
Suppose that we have a tag function $J$. Any data owner can compare tags of two workers $n_1,n_2\in[N]$, $\mathsf{tag}_{n_1}=J(X_{\mathcal{H}},X_{\mathcal{A},n_1})$ and $\mathsf{tag}_{n_2}=J(X_{\mathcal{H}},X_{\mathcal{A},n_2})$, and by Definition \ref{tag definition}, conclude that $\mathsf{tag}_{n_1}=\mathsf{tag}_{n_2}$ means $n_1$ and $n_2$ had received the same data from all data owners with high probability, i.e. $X_{k,n_1}=X_{k,n_2}$, for all $k\in[K]$. On the other hand, when $\mathsf{tag}_{n_1}\neq\mathsf{tag}_{n_2}$, it is obvious for data owners that $X_{k,n_1}\neq X_{k,n_2}$, for at least one $k\in\mathcal{A}$ (this is simply because $J$ is a function). Therefore, tags can be used by data owners to partition $\mathcal{T}$ into sets of workers who had received the same data from data owners in the first step of Algorithm \ref{Vers algo}, and have used the same initial data in their subsequent computations.

Tag functions are a type of the general \textit{fingerprinting functions}. Fingerprinting is used to map an arbitrarily large data into a small and fixed length digest as its unique identifier, for many practical purposes like avoiding the comparison of bulky data. According to \cite{hendricks2007verifying}, an $\epsilon$-fingerprinting function $fp:\mathcal{K}\times\mathbb{F}^{\delta}\rightarrow\mathbb{F}^{\gamma}$ satisfies 
\begin{align}
    \max_{\substack{d,d'\in\mathbb{F}^{\delta} \\ d\neq d'}} \textrm{Pr}\left[fp(r,d) = fp(r,d'): r \xleftarrow{R} \mathcal{K} \right] \le \epsilon,
\end{align}
where data is a length $\delta$ vector with elements in $\mathbb{F}$, and $r\in\mathcal{K}$ is a random seed, $\mathcal{K}\subset \mathbb{R}$. One of the well-known fingerprinting functions is Rabin's fingerprint \cite{rabin1981fingerprinting}, and its derivatives, that uses random polynomials in the finite field to generate the fingerprint. Another class of fingerprinting functions are cryptographic hash functions. We know that hash function $h$ is collision resistant if
\begin{align}\label{collision resistant hash}
    \textrm{Pr}\big((x_0,x_1)\leftarrow A, x_0\neq x_1: h(x_0)=h(x_1)\big)\le\epsilon,
\end{align}
where $A$ is a probabilistic polynomial adversary (This is just a rough definition. For formal definition, refer to \cite{katz2020introduction}). We cannot use a collision-resistant hash function as a tag function for two reasons.
\begin{itemize}
    \item The probability of collision is small only when $x_0$ and $x_1$ are chosen by a probabilistic polynomial adversary. However, tag function has to be resilient against information theoretic adversary.
    \item According to \eqref{tag definition}, inputs to tag function have a common part, $X_{\mathcal{H}}$, but there is no such constraint on $x_0$ and $x_1$ in \eqref{collision resistant hash}. Even if collision resistance was defined as $\textrm{Pr}\big((x_0,x_1,x_2)\leftarrow A: x_1\neq x_2, h(x_0||x_1)=h(x_0||x_2)\big)\le\epsilon$, such $h$ could not be used as tag, because adversaries could be any of the $\beta$ data owners and not necessarily the last $\beta$ data owners. Note that $J(X_{\mathcal{H}},X_{\mathcal{A},n})$ in \eqref{tag definition} is just notation exploitation, and in fact, worker $n\in\mathbb{N}$ calculates $J(X_{1,n},\dots,X_{K,n})$, in which the $\beta$ adversarial messages could be dispersed in any positions.
\end{itemize}

In the following theorem, we prove that among all functions from $\mathbb{U}^K$ to $\mathbb{U}$, there exists a tag function.
\begin{theorem}\label{random tag function theory}
There exists a tag function  $J^*:\mathbb{U}^K\rightarrow\mathbb{U}$ which satisfies \eqref{J definition}, given $|\mathbb{U}|$ is large enough.
\end{theorem}
The proof for this theorem is similar to Shannon's achievability proof for the capacity of discrete memoryless channels, and can be found in Appendix 

\begin{remark}\label{tag is not perfect}
Tags help data owners partition workers, but they cannot resolve all the ambiguity about the data that workers previously received from adversarial data owners. For example, assume that 
\begin{align*}
    \mathsf{tag}_1 = J^*(X_1,X_2^{(1)},X_3^{(1)}),\\
    \mathsf{tag}_2 = J^*(X_1,X_2^{(1)},X_3^{(1)}),\\
    \mathsf{tag}_3 = J^*(X_1,X_2^{(2)},X_3^{(1)}),\\
    \mathsf{tag}_4 = J^*(X_1,X_2^{(2)},X_3^{(2)}),
\end{align*}
where $J^*$ is a tag function. An honest data owner compares these tags and finds that $\mathsf{tag}_1=\mathsf{tag}_2$, $\mathsf{tag}_2\neq\mathsf{tag}_3$, $\mathsf{tag}_2\neq\mathsf{tag}_4$, and $\mathsf{tag}_3\neq\mathsf{tag}_4$. So it concludes that with high probability, workers $1$ and $2$ have received the same data from the adversaries, while workers $3, 4$ each have received different data. This honest data owner has no way to find out that workers $3$ and $4$ have received the same data from the second data owner, i.e. $X_{2,3}=X_{2,4}=X_2^{(2)}$, or workers $2$ and $3$ have received the same data from the third data owner, i.e. $X_{3,2}=X_{3,3}=X_3^{(1)}$. 
\end{remark}

\begin{comment}
\textcolor{red}{Lagrange coded computing scheme is helpful for data owners when the computation load on them (decoding the results of workers) is less than the case where they compute $f(X_k),k\in[K]$ themselves directly. It is known that interpolating a polynomial of degree $m\in\mathbb{N}$ has a complexity of $O(m\log^2m\log\log m)$. Each data owner interpolates a polynomial of degree $D=d(K-1)+1$ and then evaluates it at one point. Consequently, the computation complexity at each data owner is $O(D\log^2D\log\log D\dim\mathbb{V})$. Lagrange coded computing is better than direct computation by data owners when $O(D\log^2D\log\log D\dim\mathbb{V})< \textrm{comp}(f)\dim\mathbb{U}$, where $\textrm{comp}(f)$ is the computation complexity of $f$ :(}.
\end{comment}

\section{Analysis of $\textrm{Vers}(f,N,K,\beta,v,\textrm{Lagrange})$} \label{definitions section}
In this section, we introduce some tools and use them to analyze the system. These tools help us understand the system better, and will be useful in the next section when we study the fundamental limit of Vers. In the following, we define the coefficient vector, monomial vector, characteristic matrix, and relation matrix, and all of them are constructed from parts associated with every possible adversarial behavior from the point of view of the workers. In other words, the notions we introduce in this section are inclusive of all adversarial behaviors. In the converse proof of the fundamental limit of the system in the next section, we will show that in the worst-case scenario, all adversarial behaviors from the point view of workers can be present in the system, and therefore, need to be addressed.

Recall that there are at most $\beta$ adversarial data owners, and each of them  may send different messages to different workers, but honest data owners send one single message to all workers. The \textit{version vector} of a worker as we define in the following, indicates which version of each adversarial message that worker has received.

\begin{definition}[version vector]\label{version vector definition}
For worker $n\in[N]$, we define a version vector $\mathbf{v}\in[v^+]^K$ of length $K$, whose $k$-th element, if $k\in\mathcal{A}$, is in $[v]$, and denotes the version of the message of adversarial data owner $k\in\mathcal{A}$ received by worker $n$. Moreover, if $k\in\mathcal{H}$, the $k$-th element of $\mathbf{v}$ is $0$, indicating that all workers receive the same message from the honest data owner $k$.
\end{definition}
In other words, worker node $n$ receives $X_k^{(\mathbf{v}[k])}$ from the adversarial data owner $k$. Note that worker $n$ does not know $\mathbf{v}$, because it neither knows the adversaries nor the versions of the adversarial messages. Worker $n$ receives $K$ messages from data owners, without knowing anything else about the adversarial or honest data owners, and performs some computations on them. The version vector of worker $n$ is equivalent to the adversarial behavior that this worker observes.
Since there are at most $v$ different adversarial messages from each adversarial data owner, there exist $v^{\beta}$ different version vectors. This brings us to the next definition.

It is worth noting that we did not include subscript $n$ for version vector of worker $n$ in Definition \ref{version vector definition} because we will not address version vectors by the workers who received them. Rather, we use subscripts in version vectors to enumerate them, as will be clear in the following definition.

\begin{definition}[version vector set] The version vector set denoted by $\mathcal{V}\coloneqq\{\mathbf{v}_1,\mathbf{v}_2,\dots,\mathbf{v}_{v^{\beta}}\}\subset [v^+]^K$, is the set of all the $v^{\beta}$ possible version vectors.
\end{definition}
According to the system model described in Section \ref{problem formulation}, each worker computes a linear combination of the received messages, as in \eqref{W formula}. Since we are studying Vers with Lagrange encoding, each worker $n\in[N]$ whose version vector is $\mathbf{v}$, forms the following Lagrange polynomial when encoding the received data from data owners.
\begin{align}\label{q^V definition}
    q^{(\mathbf{v})}(z) \coloneqq\sum_{k\in\mathcal{A}} X_k^{(\mathbf{v}[k])}\prod_{j\neq k}\frac{z-\omega_j}{\omega_i-\omega_j}+ \sum_{k\in\mathcal{H}} X_k\prod_{j\neq k}\frac{z-\omega_j}{\omega_i-\omega_j},
\end{align}
where $\omega_k\in\mathbb{F}, k\in[K]$ are $K$ distinct elements assigned to data owners. This can be rewritten as
\begin{align}\label{encoding polynomial q}
    q^{(\mathbf{v})}(z) = \sum_{i=0}^{K-1}L_i(X_s, s\in \mathcal{H}, X_r^{(\mathbf{v}[r])}, r\in \mathcal{A},f(.))z^i,
\end{align}
where $L_0,\dots,L_{K-1}:\mathbb{F}^{K}\rightarrow\mathbb{F}$ are linear maps. Then, $f(q^{(\mathbf{v})}(z))$ can be rewritten as
\begin{align}\label{foq expansion}
    f(q^{(\mathbf{v})}(z)) = \sum_{i=0}^{d(K-1)}u_i(X_s, s\in \mathcal{H}, X_r^{(\mathbf{v}[r])}, r\in \mathcal{A},f(.))z^i,
\end{align}
where $u_0,\dots,u_{d(K-1)}:\mathbb{F}^{K}\rightarrow\mathbb{F}$ are polynomials of degree $d$ of $X_r^{(\mathbf{v}[r])}, r\in\mathcal{A}$ and $X_s,s\in\mathcal{H}$. Worker $n\in[N]$ evaluates $f(q^{(\mathbf{v})}(z))$ at $\alpha_n$, i.e. $f(W_n)=f(q^{(\mathbf{v})}(\alpha_n))$, where $\alpha_n, n\in[N]$ are $n$ elements from $\mathbb{F}$, assigned to workers. In the following, we define the coefficient vector and coefficient sub-vector, so that we can express the computations of workers as multiplications of matrices.
\begin{definition}[coefficient vector]\label{coefficient vector definition}
Let the version vector set be $\mathcal{V} = \{\mathbf{v}_1,\dots,\mathbf{v}_{v^{\beta}}\}$. The coefficient vector for a  $\textrm{Vers}(f,N,K,\beta,v,\textrm{Lagrange})$ system is defined as
\begin{align}\label{U definition}
    \mathbf{U}\big(X_s, s\in\mathcal{H}, X_r^{(i)}, r\in\mathcal{A},i\in[v],f(.)\big) \coloneqq 
    \left[ \begin{array}{c}
    \mathbf{u}_{\mathbf{v}_1}(X_r^{(\mathbf{v}_1[r])},X_s, r\in\mathcal{A},s\in\mathcal{H},f)  \\ \hdashline \vdots \\ \hdashline \mathbf{u}_{\mathbf{v}_{v^{\beta}}}(X_r^{(\mathbf{v}_{v^{\beta}}[r])},X_s, r\in\mathcal{A},s\in\mathcal{H},f)
    \end{array} \right],
\end{align}
where $\mathbf{u}_{\mathbf{v}}\big(X_s, s\in \mathcal{H}, X_r^{(\mathbf{v}[r])}, r\in \mathcal{A},f(.)\big), \mathbf{v}\in\mathcal{V}$ is a coefficient sub-vector defined as 
\begin{align}\label{u (subvector) definition}
    \mathbf{u}_{\mathbf{v}}\big(X_s, s\in \mathcal{H}, X_r^{(\mathbf{v}[r])}, r\in \mathcal{A},f(.)\big) \coloneqq \left[\begin{array}{c}
    u_{d(K-1)}\big(X_s, s\in \mathcal{H}, X_r^{(\mathbf{v}[r])}, r\in \mathcal{A},f(.)\big) \\
    \vdots \\
    u_{0}\big(X_s, s\in \mathcal{H}, X_r^{(\mathbf{v}[r])}, r\in \mathcal{A},f(.)\big)
    \end{array}\right], \quad \mathbf{v}\in\mathcal{V}.
\end{align}
\end{definition}
The length of each coefficient sub-vector is $d(K-1)+1$, so the length of coefficient vector $\mathbf{U}$ is $v^{\beta}(d(K-1)+1)$. As mentioned before, the coefficient vector contains all of the $v^{\beta}$ possible adversarial behaviors. Using the definition above, it is clear that for worker $n\in[N]$ whose version vector is $\mathbf{v}\in\mathcal{V}$,
\begin{align}
    f(W_n) = \left[\alpha^{d(K-1)}~\dots~\alpha~ 1 \right]\mathbf{u}_{\mathbf{v}}\big(X_s, s\in \mathcal{H}, X_r^{(\mathbf{v}[r])}, r\in \mathcal{A},f(.)\big).
\end{align}
Therefore, the coefficient vector helps us express the computations of the workers properly. This motivates us to study the coefficient vector more in-depth. The following example clarifies the above definition.
\begin{example}\label{example 1}
Consider the following parameters.
\begin{center}
\begin{tabular}{ |c|c|c|c|c|c|c|c| } 
 \hline
 $K$ & $N$ & $\beta$ & $v$ & $\mathcal{A}$ & $\mathcal{H}$ & $(\omega_1,\omega_2,\omega_3)$ & $f(x)$  \\ \hline
 $3$ & $3$ & $1$ & $2$ & $\{1,2\}$ & $\{3\}$ & $(1,2,3)$ & $x^2$ \\ \hline
\end{tabular}
\end{center}
Since $v=2$, the version vector set is $\mathcal{V} = \{[1,1,0],[1,2,0],[2,1,0],[2,2,0]\}$. The Lagrange polynomial for the version vector $[i_1,i_2,0]\in\mathcal{V}$ is
\begin{align*}
    q^{([i_1,i_2,0])}(z)&=\frac{(z-2)(z-3)}{2}X^{(i_1)}_1 + \frac{(z-1)(z-3)}{-1}X_2^{(i_2)} +\frac{(z-1)(z-2)}{2}X_3= z^2(\frac{X_1^{(i_1)}}{2}-X_2^{(i_2)}+\frac{X_3}{2}) \nonumber \\
    & + z(-\frac{5X_1^{(i_1)}}{2} +2X_2^{(i_2)} -\frac{3X_3}{2})+(3X_1^{(i_1)}-3X_2^{(i_2)}+X_3),~ [i_1,i_2,0]\in\mathcal{V}.
\end{align*}
Consequently,
\begin{align}\label{f o u in example 1}
f(q^{([i_1,i_2,0])}(z)) &=\bigg(\frac{(X_1^{(i_1)})^2}{4}-X_1^{(i_1)}X_2^{(i_2)}+\frac{X_1^{(i_1)}X_3}{2}+(X_2^{(i_2)})^2  -X_2^{(i_2)}X_3+\frac{X_3^2}{4}\bigg)z^4 + \bigg(\frac{-5(X_1^{(i_1)})^2}{2} \nonumber \\
& +9X_1^{(i_1)}X_2^{(i_2)} -4X_1^{(i_1)}X_3 -8(X_2^{(i_2)})^2+7X_2^{(i_2)}X_3-\frac{3X_3^2}{2}\bigg)z^3 + \bigg(\frac{37(X_1^{(i_1)})^2}{4}-29X_1^{(i_1)}X_2^{(i_2)}   \nonumber \\
& +\frac{23X_1^{(i_1)}X_3}{2}+22(X_2^{(i_2)})^2 -17X_2^{(i_2)}X_3 +\frac{13X_3^2}{4}\bigg)z^2 +\bigg(-15(X_1^{(i_1)})^2 +39X_1^{(i_1)}X_2^{(i_2)}   \nonumber \\
& -14X_1^{(i_1)}X_3-24(X_2^{(i_2)})^2+17X_2^{(i_2)}X_3 -3X_3\bigg)z +\bigg(9(X_1^{(i_1)})^2 -18X_1^{(i_1)}X_2^{(i_2)} +6X_1^{(i_1)}X_3 \nonumber \\ 
& +9(X_2^{(i_2)})^2-6X_2^{(i_2)}X_3+X_3^2\bigg).
\end{align}
Therefore,
\begin{align}\label{coefficient sub vector example}
    \mathbf{u}_{[i_1,i_2,0]}(X_1^{(i_1)},X_2^{(i_2)},X_3,f) = \left[ \begin{array}{c}
        \frac{(X_1^{(i_1)})^2}{4}-X_1^{(i_1)}X_2^{(i_2)}+\frac{X_1^{(i_1)}X_3}{2}+(X_2^{(i_2)})^2  -X_2^{(i_2)}X_3+\frac{X_3^2}{4} \\
        \frac{-5(X_1^{(i_1)})^2}{2}+9X_1^{(i_1)}X_2^{(i_2)}-4X_1^{(i_1)}X_3 -8(X_2^{(i_2)})^2+7X_2^{(i_2)}X_3-\frac{3X_3^2}{2} \\
        \frac{37(X_1^{(i_1)})^2}{4}-29X_1^{(i_1)}X_2^{(i_2)}  +\frac{23X_1^{(i_1)}X_3}{2}+22(X_2^{(i_2)})^2 -17X_2^{(i_2)}X_3 +\frac{13X_3^2}{4} \\
        -15(X_1^{(i_1)})^2 +39X_1^{(i_1)}X_2^{(i_2)}-14X_1^{(i_1)}X_3 -24(X_2^{(i_2)})^2+17X_2^{(i_2)}X_3-3X_3 \\
        9(X_1^{(i_1)})^2 -18X_1^{(i_1)}X_2^{(i_2)} +6X_1^{(i_1)}X_3+9(X_2^{(i_2)})^2-6X_2^{(i_2)}X_3+X_3^2
    \end{array}
    \right],
\end{align}
and
\begin{align}\label{coefficient vector example}
\mathbf{U}(X_1^{(i_1)},X_2^{(i_2)},i_1,i_2\in[2],X_3,f) = \left[ \begin{array}{c}
     \mathbf{u}_{[1,1,0]}(X_1^{(1)},X_2^{(1)},X_3,f) \\
     \mathbf{u}_{[1,2,0]}(X_1^{(1)},X_2^{(2)},X_3,f) \\
     \mathbf{u}_{[2,1,0]}(X_1^{(2)},X_2^{(1)},X_3,f) \\
     \mathbf{u}_{[2,2,0]}(X_1^{(2)},X_2^{(2)},X_3,f) 
\end{array} \right]
\end{align}
\end{example}
\begin{remark}\label{order in coefficient vector}
In Definition \ref{coefficient vector definition}, we implied a particular order on the coefficient sub-vectors in \eqref{U definition}. The first and topmost sub-vector in the coefficient vector is $\mathbf{u}_{\mathbf{v}_1}$, the second is $\mathbf{u}_{\mathbf{v}_2}$, and so forth, until $\mathbf{u}_{\mathbf{v}_{v^{\beta}}}$. In Example \ref{example 1}, we chose $\mathbf{v}_1=[1,1,0]$, $\mathbf{v}_2=[1,2,0]$, $\mathbf{v}_3=[2,1,0]$, and $\mathbf{v}_4=[2,2,0]$. We could have used another assignment for $\mathbf{v}_1,\mathbf{v}_2,\mathbf{v}_3,\mathbf{v}_4$, as different assignments only lead to different permutations of the same coefficient vector. Therefore, the order of coefficient sub-vectors in coefficient vector is optional but fixed.
\end{remark}
Every element of the coefficient vector is a summation of some monomials. For example, in equation \eqref{coefficient sub vector example} from Example \ref{example 1}, $\mathbf{u}_{[i_1,i_2,0]}(X_1^{(i_1)},X_2^{(i_2)},X_3,f)$  shows that elements of the coefficient vector are made of some monomials. Therefore, in order to understand the coefficient vector better, we define a new vector that contains the monomials that appear in the coefficient vector. Before defining the monomial vector, we need to define the degree set of a function. The degree set of a single variable function $f:\mathbb{U}\rightarrow\mathbb{V}$, denoted by $\mathcal{D}(f)$, contains the degrees of the monomials in $f$. For example, $\mathcal{D}(f)=\{0,1,3\}$ for $f(x)=x^3+x+1$.
\begin{definition}[monomial vector]
The \textit{monomial vector} for a $\textrm{Vers}(f,N,K,\beta,v,\textrm{Lagrange})$ system contains all degree $e\in\mathcal{D}_f$ monomials of $X_k, k\in \mathcal{H}$ and $X_k^{(i)}, k\in\mathcal{A}, i\in[v]$, and is defined as 
\begin{align}\label{monomial vector formula}
    \mathbf{X}\big(X_s, s\in\mathcal{H}, X_r^{(i)}, r\in\mathcal{A},i\in[v],f(.)\big) \coloneqq \left[\prod_{k\in\mathcal{A}}(X_k^{(\mathbf{v}[k])})^{i_k}\prod_{k\in\mathcal{H}}(X_{k})^{j_k}, ~\textrm{s.t.} \sum_{k\in\mathcal{A}}i_k+\sum_{k\in\mathcal{H}}j_k\in \mathcal{D}(f), \textrm{and} ~\mathbf{v}\in\mathcal{V} \right]^\intercal.
\end{align}
\end{definition}
In the following example, we determine the monomial vector for Vers system introduced in Example \ref{example 1}.
\begin{example}\label{example 2, monomial of example 1}
Recall that in Example \ref{example 1}, $f(x)=x^2$, and $\mathcal{D}(f)=\{2\}$. Therefore, the monomial vector contains all degree $2$ monomials of $X_k, k\in \mathcal{H}$ and $X_k^{(i)}, k\in\mathcal{A}, i\in[v]$. In particular, 
\begin{itemize}
    \item for $i_1=2, i_2=0, i_3=0$, we get the monomials $(X_1^{(1)})^2,~ (X_1^{(2)})^2$,
    \item for $i_1=0, i_2=2, i_3=0$, we get the monomials $(X_2^{(1)})^2,~ (X_2^{(2)})^2$,
    \item for $i_1=0, i_2=0, i_3=2$, we get the monomial $(X_3)^2$,
    \item for $i_1=1, i_2=1, i_3=0$, we get the monomials $X_1^{(1)}X_2^{(1)}, ~ X_1^{(2)}X_2^{(1)}, X_1^{(1)}X_2^{(2)}, ~ X_1^{(2)}X_2^{(2)}$,
    \item for $i_1=1, i_2=0, i_3=1$, we get the monomials $X_1^{(1)}X_3, ~ X_1^{(2)}X_3$,
    \item for $i_1=0, i_2=1, i_3=1$, we get the monomials $X_2^{(1)}X_3, ~ X_2^{(2)}X_3$.
\end{itemize}
Therefore, 
\begin{align}\label{monomial vector example}
    \mathbf{X}\big(X_1^{(i_1)},X_2^{(i_2)},X_3,i_1,i_2\in[2],f(x)=x^2\big) =& \bigg[(X_1^{(1)})^2,~ (X_1^{(2)})^2, ~ (X_2^{(1)})^2,~ (X_2^{(2)})^2, ~ (X_3)^2, X_1^{(1)}X_2^{(1)}, X_1^{(2)}X_2^{(1)},\nonumber\\ 
    &X_1^{(1)}X_2^{(2)}, ~ X_1^{(2)}X_2^{(2)}, 
    X_1^{(1)}X_3, ~ X_1^{(2)}X_3, 
     X_2^{(1)}X_3, ~ X_2^{(2)}X_3\bigg]^\intercal.
\end{align}
\end{example}
\begin{remark}
Note that the order of the monomials in monomial vector is optional, but fixed in all subsequent equations. For example, we could have formed $\mathbf{X}\big(X_1^{(i_1)},X_2^{(i_2)},X_3,i_1,i_2\in[2],f(x)=x^2\big)$ in \eqref{monomial vector example} as
\begin{align*}
    \biggr[(X_1^{(2)})^2,~ (X_1^{(1)})^2, ~ (X_2^{(2)})^2,~ (X_2^{(1)})^2, ~ (X_3)^2, & ~ X_1^{(2)}X_2^{(1)},~ X_1^{(1)}X_2^{(1)}, ~  X_1^{(2)}X_2^{(2)}, X_1^{(1)}X_2^{(2)},  \\
    & X_1^{(1)}X_3,
    ~ X_1^{(2)}X_3,  X_2^{(1)}X_3, ~ X_2^{(2)}X_3\biggr]^\intercal.
\end{align*}
However, we chose the order in \eqref{monomial vector example}, and will keep using that in the follow-up of Example \ref{example 2, monomial of example 1}.
\end{remark}
\begin{lemma}
The length of monomial vector $\mathbf{X}\big(X_r^{(i)},X_s, r\in\mathcal{A},s\in\mathcal{H},i\in[v]\big)$ in a $\textrm{Vers}(f,N,K,\beta,v,\textrm{Lagrange})$ system is
\begin{align}\label{L definition}
    \mathsf{Len}(\mathbf{X}) \coloneqq \sum_{e\in\mathcal{D}(f)}\sum_{m=0}^{\min(e,\beta)}\binom{\beta}{m}\binom{K-\beta-1+e}{e-m}v^m
\end{align}
\end{lemma}
\begin{proof}
For each $e\in\mathcal{D}(f)$, we need to count the number of distinct monomials $\prod_{k\in\mathcal{A}}(X_k^{(\mathbf{v}[k])})^{i_k}\prod_{k\in\mathcal{H}}(X_{k})^{j_k}$, that $\sum_{k\in\mathcal{A}}i_k+\sum_{k\in\mathcal{H}}j_k=e$ and $\mathbf{v}\in\mathcal{V}$. We count as following:
\begin{itemize}
    \item We choose a subset $\mathcal{M}\in\mathcal{A}$, $|\mathcal{M}|=m$, of adversarial messages to be in the monomial, so $m\le \beta$ and $m\le e$, thus $0\le m\le \min(e,\beta)$. This makes the term $\binom{\beta}{m}$ in \eqref{L definition}.
    \item There are $v^m$ different version values for the $m$ adversarial messages in the monomial, which makes the term $v^m$ in \eqref{L definition}.
    \item We need to count the number of possible cases for the powers of the chosen $m$ adversarial messages and the $K-\beta$ honest messages. This is the number of solutions for $\sum_{k\in\mathcal{M}}i_k+\sum_{k\in\mathcal{H}}j_k=e$, where $1\le i_k$, for $k\in\mathcal{M}$, and $0\le j_k$, for $k\in\mathcal{H}$. From combinatorics, we know that the number of such solutions is $\binom{K-\beta-1+e}{e-m}$. This completes the proof.
\end{itemize}
\end{proof}
Now that the monomial vector is defined, we can express each element of the coefficient vector as a linear combination of the monomials in the monomial vector. We do so by defining the characteristic matrix.
\begin{definition}[characteristic matrix] \label{M definition}
The \textit{characteristic matrix} of a $\textrm{Vers}(f,N,K,\beta,v,\textrm{Lagrange})$ system, denoted by $\mathbf{M}(\mathcal{H},\mathcal{A},v,f)$ illustrates the relation between the coefficient vector and the monomial vector, i.e.
\begin{align}
    \mathbf{U}\big(X_s, s\in\mathcal{H}, X_r^{(i)}, r\in\mathcal{A},i\in[v],f(.)\big)=\mathbf{M}\big(\mathcal{H},\mathcal{A},v,f\big)\mathbf{X}\big(X_s, s\in\mathcal{H}, X_r^{(i)}, r\in\mathcal{A},i\in[v],f(.)\big).
\end{align}
\end{definition}
The dimension of the characteristic matrix is $\big(v^{\beta}(d(K-1)+1)\big)\times \mathsf{Len}(\mathbf{X})$. It is worth noting that similar to the coefficient vector, the characteristic matrix contains all the $v^{\beta}$ adversarial behaviors.
\begin{example}\label{example: M of example 1}
For Vers system introduced in Example \ref{example 1}, we can easily form $\mathbf{M}\big(\mathcal{H}=\{3\},\mathcal{A}=\{1,2\},v=2,f(x)=x^2\big)$ by using \eqref{coefficient sub vector example}, \eqref{coefficient vector example}, and also \eqref{monomial vector example} from Example \ref{example 2, monomial of example 1}. The resulting characteristic matrix is shown in \eqref{M for example 1}.
\end{example}
We can decompose the characteristic matrix as
\begin{align}\label{M decomposition}
    \mathbf{M}\big(\mathcal{H},\mathcal{A},v,f\big) =  \left[ \begin{array}{c}
    \mathbf{M}_{\mathbf{v}_1}\big(\mathcal{H},\mathcal{A},f\big)  \\ \hdashline \vdots \\ \hdashline \mathbf{M}_{\mathbf{v}_{v^{\beta}}}\big(\mathcal{H},\mathcal{A},f\big)
    \end{array} \right],
\end{align}
where $\{\mathbf{v}_1,\dots,\mathbf{v}_{v^{\beta}}\}=\mathcal{V}$. The dimension of each sub-matrix $\mathbf{M}_{\mathbf{v}}\big(\mathcal{H},\mathcal{A},f\big)$, $\mathbf{v}\in\mathcal{V}$, is $\big(d(K-1)+1\big)\times \mathsf{Len}(\mathbf{X})$. In fact, the sub-matrix $\mathbf{M}_{\mathbf{v}}$ is dedicated to coefficient sub-vector $\mathbf{u}_{\mathbf{v}}\big(X_s, s\in \mathcal{H}, X_r^{(\mathbf{v}[r])}, r\in \mathcal{A},f(.)\big)$, and expresses its elements as linear combinations of the element of the monomial vector. Consequently, the order of the version vectors that we had chosen in the coefficient vector, as discussed in Remark \ref{order in coefficient vector}, determines the order of the sub-matrices in the characteristic matrix. We emphasize that this order is optional but fixed.

\begin{figure*}
    \centering
%     % \includegraphics{}
%     % \caption{Caption}
%     % \label{fig:my_label}
\begin{align}\label{M for example 1} \hspace*{-1.5cm} 
\left[\begin{array}{c}
     u_4(X_1^{(1)},X_2^{(1)},X_3) \\
     u_3(X_1^{(1)},X_2^{(1)},X_3) \\
     u_2(X_1^{(1)},X_2^{(1)},X_3) \\
     u_1(X_1^{(1)},X_2^{(1)},X_3) \\
     u_0(X_1^{(1)},X_2^{(1)},X_3) \\ \hdashline
     u_4(X_1^{(1)},X_2^{(2)},X_3) \\
     u_3(X_1^{(1)},X_2^{(2)},X_3) \\
     u_2(X_1^{(1)},X_2^{(2)},X_3) \\
     u_1(X_1^{(1)},X_2^{(2)},X_3) \\
     u_0(X_1^{(2)},X_2^{(2)},X_3) \\ \hdashline
     u_4(X_1^{(2)},X_2^{(1)},X_3) \\
     u_3(X_1^{(2)},X_2^{(1)},X_3) \\
     u_2(X_1^{(2)},X_2^{(1)},X_3) \\
     u_1(X_1^{(2)},X_2^{(1)},X_3) \\
     u_0(X_1^{(2)},X_2^{(1)},X_3) \\ \hdashline
     u_4(X_1^{(2)},X_2^{(2)},X_3) \\
     u_3(X_1^{(2)},X_2^{(2)},X_3) \\
     u_2(X_1^{(2)},X_2^{(2)},X_3) \\
     u_1(X_1^{(2)},X_2^{(2)},X_3) \\
     u_0(X_1^{(2)},X_2^{(2)},X_3)
\end{array}\right]= \left[\begin{array}{cc:cc:c:cccc:cc:cc}
     \frac{1}{4} & 0 & 1 & 0 & \frac{1}{4} & -1 & 0 & 0 & 0 & \frac{1}{2} & 0 & -1 & 0 \\
     -\frac{5}{2} & 0 & -8 & 0 & -\frac{3}{2} & 9 & 0 & 0 & 0 & -4 & 0 & 7 & 0 \\
     \frac{37}{4} & 0 & 22 & 0 & \frac{13}{4} & -29 & 0 & 0 & 0 & \frac{23}{2} & 0 & -17 & 0 \\
     -15 & 0 & -24 & 0 & -3 & 39 & 0 & 0 &  0 & -14 & 0 & -17 & 0 \\
     9 & 0 & 9 & 0 & 1 & -18 & 0 & 0 & 0 & 6 & 0 & -6 & 0 \\ \hdashline
     \frac{1}{4} & 0 & 0 & 1 & \frac{1}{4} & 0 & 0 & -1 & 0 & \frac{1}{2} & 0 & 0 & -1 \\
     -\frac{5}{2} & 0 & 0 & -8 & -\frac{3}{2} & 0 & 0 & 9 & 0 & -4 & 0 & 0 & 7 \\
     \frac{37}{4} & 0 & 0 & 22 & \frac{13}{4} & 0 & 0 & -29 & 0 & \frac{23}{2} & 0 & 0 & -17 \\
     -15 & 0 & 0 & -24 & -3 & 0 & 0 & 39 & 0 & -14 & 0 & 0 & 17 \\
     9 & 0 & 0 & 9 & 1 & 0 & 0 & -18 & 0 & 6 & 0 & 0 & -6 \\ \hdashline
     0 & \frac{1}{4} & 1 & 0 & \frac{1}{4} & 0 & -1 & 0 & 0 & 0 & \frac{1}{2} & -1 & 0 \\
     0 & -\frac{5}{2} & -8 & 0 & -\frac{3}{2} & 0 & 9 & 0 & 0 & 0 & -4 & 7 & 0 \\
     0 & \frac{37}{4} & 22 & 0 & \frac{13}{4} & 0 & -29 & 0 & 0 & 0 & \frac{23}{2} & -17 & 0 \\
     0 & -15 & -24 & 0 & -3 & 0 & 39 & 0 & 0 & 0 & -14 & 17 & 0 \\
     0 & 9 & 9 & 0 & 1 & 0 & -18 & 0 & 0 & 0 & 6 & -6 & 0 \\ \hdashline
     0 & \frac{1}{4} & 0 & 1 & \frac{1}{4} & 0 & 0 & 0 & -1 & 0 & \frac{1}{2} & 0 & -1 \\
     0 & -\frac{5}{2} & 0 & -8 & -\frac{3}{2} & 0 & 0 & 0 & 9 & 0 & -4 & 0 & 7 \\
     0 & \frac{37}{4} & 0 & 22 & \frac{13}{4} & 0 & 0 & 0 & -29 & 0 & \frac{23}{2} & 0 & -17 \\
     0 & -15 & 0 & -24 & -3 & 0 & 0 & 0 & 39 & 0 & -14 & 0 & 17 \\
     0 & 9 & 0 & 9 & 1 & 0 & 0 & 0 & -18 & 0 & 6 & 0 & -6
\end{array}\right]\left[\begin{array}{c}
     (X_1^{(1)})^2 \\ (X_1^{(2)})^2 \\ \hdashline
     (X_2^{(1)})^2 \\ (X_2^{(2)})^2 \\ \hdashline
     (X_3)^2 \\ \hdashline
     X_1^{(1)}X_2^{(1)} \\ X_1^{(2)}X_2^{(1)} \\ X_1^{(1)}X_2^{(2)} \\ X_1^{(2)}X_2^{(2)} \\ \hdashline X_1^{(1)}X_3 \\ X_1^{(2)}X_3 \\ \hdashline
     X_2^{(1)}X_3 \\ X_2^{(2)}X_3
\end{array}\right].
\end{align}
\end{figure*}
\begin{definition}(relation matrix)\label{P definition}
The \textit{relation matrix} in a $\textrm{Vers}(f,N,K,\beta,v,\textrm{Lagrange})$ system, which we denote by $\mathbf{P}\big(\mathcal{H},\mathcal{A},v,f\big)$, consists of basis of the left null space of  $\mathbf{M}\big(\mathcal{H},\mathcal{A},v,f\big)$, i.e.
\begin{align}\label{PM=0}
    \mathbf{P}\big(\mathcal{H},\mathcal{A},v,f\big)\mathbf{M}\big(\mathcal{H},\mathcal{A},v,f\big) = \mathbf{0}.
\end{align}
\end{definition}
The dimension of the relation matrix is $\big(v^{\beta}(d(K-1)+1)-\textrm{rank}(\mathbf{M})\big) \times \big(v^{\beta}(d(K-1)+1)\big)$. According to Definition \ref{M definition}, $\mathbf{P}\big(\mathcal{H},\mathcal{A},v,f\big)\mathbf{M}\big(\mathcal{H},\mathcal{A},v,f\big) = \mathbf{0}$ is equivalent to
\begin{align}\label{PU=0}
    \mathbf{P}\big(\mathcal{H},\mathcal{A},v,f\big)\mathbf{U}\big(X_s, s\in\mathcal{H}, X_r^{(i)}, r\in\mathcal{A},i\in[v],f(.)\big) = \mathbf{0}.
\end{align}
This is why we have named $\mathbf{P}\big(\mathcal{H},\mathcal{A},v,f\big)$ as relation matrix, because it reveals the linear relations of the elements of the coefficient vector. Similar to the characteristic matrix, we can decompose the relation matrix as
\begin{align}\label{P decomposition}
    \mathbf{P}\big(\mathcal{H},\mathcal{A},v,f\big) = \left[
    \begin{array}{c:c:c}
       \mathbf{P}_{\mathbf{v}_1}\big(\mathcal{H},\mathcal{A},f\big) & \dots & \mathbf{P}_{\mathbf{v}_{v^{\beta}}}\big(\mathcal{H},\mathcal{A},f\big)
    \end{array}
    \right],
\end{align}
where $\{\mathbf{v}_1,\dots,\mathbf{v}_{v^{\beta}}\}=\mathcal{V}$. The dimension of each sub-matrix $\mathbf{P}_{\mathbf{v}}\big(\mathcal{H},\mathcal{A},f\big)$, $\mathbf{v}\in\mathcal{V}$, is $\big(v^{\beta}(d(K-1)+1)-\textrm{rank}(\mathbf{M})\big) \times (d(K-1)+1)$. Due to \eqref{M decomposition} and \eqref{PM=0},
\begin{align*}
    \mathbf{P}_{\mathbf{v}_1}\mathbf{M}_{\mathbf{v}_1}+\dots+\mathbf{P}_{\mathbf{v}_{v^{\beta}}}\mathbf{M}_{\mathbf{v}_{v^{\beta}}} = \mathbf{0},
\end{align*}
and due to \eqref{M definition},
\begin{align*}
    \mathbf{P}_{\mathbf{v}_1}\mathbf{u}_{\mathbf{v}_1}+\dots+\mathbf{P}_{\mathbf{v}_{v^{\beta}}}\mathbf{u}_{\mathbf{v}_{v^{\beta}}} = \mathbf{0}.
\end{align*}
\begin{comment}
The following lemma presents a property of the relation matrix and helps us understand this matrix better. 
\begin{lemma}\label{submatrices of P are full rank}
In a $\textrm{Vers}(f,N,K,\beta,v,\textrm{Lagrange})$ system, where $\textrm{deg}(f)<\beta$, all $v^{\beta}$ sub-matrices of $\mathbf{P}(\mathcal{H},\mathcal{A},v,f)$ as defined in \eqref{P decomposition}, i.e. $\mathbf{P}_{\mathbf{v}_1}\big(\mathcal{H},\mathcal{A},f\big), \dots,\mathbf{P}_{\mathbf{v}_{v^{\beta}}}\big(\mathcal{H},\mathcal{A},f\big)$ are all full column rank.
\end{lemma}
The proof for this lemma can be found in Appendix \ref{proof of submatrices of P are full rank}.
\end{comment}
Now, we introduce \textit{effective} and \textit{non-effective} permutations of the relation matrix.
\begin{definition}(effective permutations of the relation matrix)\label{permutation of P definition}
Suppose that 
\begin{align*}
    \mathbf{P}\big(\mathcal{H},\mathcal{A},v,f\big) = \left[
    \begin{array}{c:c:c}
       \mathbf{P}_{\mathbf{v}_1}\big(\mathcal{H},\mathcal{A},f\big) & \dots & \mathbf{P}_{\mathbf{v}_{v^{\beta}}}\big(\mathcal{H},\mathcal{A},f\big)
    \end{array}
    \right]
\end{align*}
is the relation matrix of a $\textrm{Vers}(f,N,K,\beta,v,\textrm{Lagrange})$ system. A permutation $\Pi:[v^{\beta}]\rightarrow[v^{\beta}]$ of the relation matrix is $\mathbf{P}^{\Pi}\big(\mathcal{H},\mathcal{A},v,f\big) = \left[
    \begin{array}{c:c:c}
       \mathbf{P}_{\mathbf{v}_{\Pi(1)}}\big(\mathcal{H},\mathcal{A},f\big) & \dots & \mathbf{P}_{\mathbf{v}_{\Pi(v^{\beta})}}\big(\mathcal{H},\mathcal{A},f\big)
    \end{array}
    \right]$, and is effective if $\mathbf{P}^{\Pi}\mathbf{M}\neq \mathbf{0}$, and non-effective if $\mathbf{P}^{\Pi}\mathbf{M}= \mathbf{0}$.
\end{definition}
The intuition behind this definition is the special structure of the characteristic matrix, which is also evident in \eqref{M for example 1}. This special structure will get clear in the following example.
\begin{example}\label{effective permutation of example 1}
Recall that for Vers system introduced in Example \eqref{example 1}, we had set $\mathbf{v}_1=[1,1,0]$, $\mathbf{v}_2=[1,2,0]$, $\mathbf{v}_3=[2,1,0],\mathbf{v}_4=[2,2,0]$, and $\mathbf{P}\big(\mathcal{H}=\{3\},\mathcal{A}=\{1,2\},v=2,f(x)=x^2\big)=[\mathbf{P}_{\mathbf{v}_1}~ \mathbf{P}_{\mathbf{v}_2}~ \mathbf{P}_{\mathbf{v}_3}~ \mathbf{P}_{\mathbf{v}_4}]$. We show that in this Vers system,
there are $3$ non-effective permutations of $\mathbf{P}$, say $\Pi_1,\Pi_2,$ and $\Pi_3$, as following.
\begin{align}
    \mathbf{P}^{\Pi_1} &= [\mathbf{P}_{\mathbf{v}_4}~ \mathbf{P}_{\mathbf{v}_3}~ \mathbf{P}_{\mathbf{v}_2}~ \mathbf{P}_{\mathbf{v}_1}] \label{pi_1 in example 1}\\
    \mathbf{P}^{\Pi_2} &= [\mathbf{P}_{\mathbf{v}_3} ~\mathbf{P}_{\mathbf{v}_4} ~\mathbf{P}_{\mathbf{v}_1} ~\mathbf{P}_{\mathbf{v}_2}] \label{pi_2 in example 1}\\
    \mathbf{P}^{\Pi_3} &= [\mathbf{P}_{\mathbf{v}_2}~ \mathbf{P}_{\mathbf{v}_1}~ \mathbf{P}_{\mathbf{v}_4}~ \mathbf{P}_{\mathbf{v}_3}] \label{pi_3 in example 1}
\end{align}
In other words, $\mathbf{P}^{\Pi_1}\mathbf{M}=\mathbf{P}^{\Pi_2}\mathbf{M}=\mathbf{P}^{\Pi_3}\mathbf{M}=\mathbf{0}$, where $\mathbf{M}$ is given in \eqref{M for example 1}. Let us consider $\Pi_1$ and show that $\mathbf{P}^{\Pi_1}\mathbf{M}=\mathbf{0}$. The arguments for the other two permutations are similar. We denote a column $i$, $1\le i\le 13$, of $\mathbf{M}$ as $\mathbf{c_i}=[c_{i,1}^\intercal~ c_{i,2}^\intercal~ c_{i,3}^\intercal~ c_{i,4}^\intercal]^\intercal$, where the length of $c_{i,1},c_{i,2},c_{i,3},c_{i,4}$ is $5$.
We know that $\mathbf{P}c_i=\mathbf{P}_{\mathbf{v}_1}c_{i,1}+\mathbf{P}_{\mathbf{v}_2}c_{i,2}+\mathbf{P}_{\mathbf{v}_3}c_{i,3}+\mathbf{P}_{\mathbf{v}_4}c_{i,4}=\mathbf{0}$. We consider all  $1\le i\le 13$ in the following.
\begin{itemize}
\item $i=1, i=2$ \\
A close look at the first two columns of $\mathbf{M}$ in \eqref{M for example 1} reveals that $c_{1,1}=c_{1,2}=c_{2,3}=c_{2,4}\coloneqq c^*_{\{1,2\}}$, and $c_{1,3}=c_{1,4}=c_{2,1}=c_{2,2}=\mathbf{0}$. Therefore 
\begin{align*}
    \mathbf{P}\mathbf{c}_1 = \mathbf{P}_{\mathbf{v}_1}c_{1,1}+\mathbf{P}_{\mathbf{v}_2}c_{1,2}+\mathbf{P}_{\mathbf{v}_3}c_{1,3}+\mathbf{P}_{\mathbf{v}_4}c_{1,4}=\mathbf{0},
\end{align*}
is equivalent to 
\begin{align}\label{p1 c1}
    (\mathbf{P}_{\mathbf{v}_1}+\mathbf{P}_{\mathbf{v}_2})c^*_{\{1,2\}} = \mathbf{0},
\end{align}
and
\begin{align*}
    \mathbf{P}\mathbf{c}_2 =\mathbf{P}_{\mathbf{v}_1}c_{2,1}+\mathbf{P}_{\mathbf{v}_2}c_{2,2}+\mathbf{P}_{\mathbf{v}_3}c_{2,3}+\mathbf{P}_{\mathbf{v}_4}c_{2,4}=\mathbf{0},
\end{align*}
is equivalent to 
\begin{align}\label{p1 c2}
    (\mathbf{P}_{\mathbf{v}_3}+\mathbf{P}_{\mathbf{v}_4})c^*_{\{1,2\}} = \mathbf{0}.
\end{align}
Consequently, for the permuted relation matrix in \eqref{pi_1 in example 1},
\begin{align*}
    \mathbf{P}^{\Pi_1}\mathbf{c}_1 = \mathbf{P}_{\mathbf{v}_4}c_{1,1}+\mathbf{P}_{\mathbf{v}_3}c_{1,2}+\mathbf{P}_{\mathbf{v}_2}c_{1,3}+\mathbf{P}_{\mathbf{v}_1}c_{1,4} = (\mathbf{P}_{\mathbf{v}_4}+\mathbf{P}_{\mathbf{v}_3})c^*_{\{1,2\}},
\end{align*}
and due to \eqref{p1 c2}, we conclude that $\mathbf{P}^{\Pi_1}\mathbf{c}_1 =\mathbf{0}$. Similarly,
\begin{align*}
    \mathbf{P}^{\Pi_1}\mathbf{c}_2 = \mathbf{P}_{\mathbf{v}_4}c_{2,1}+\mathbf{P}_{\mathbf{v}_3}c_{2,2}+\mathbf{P}_{\mathbf{v}_2}c_{2,3}+\mathbf{P}_{\mathbf{v}_1}c_{2,4} = (\mathbf{P}_{\mathbf{v}_2}+\mathbf{P}_{\mathbf{v}_1})c^*_{\{1,2\}},
\end{align*}
and due to \eqref{p1 c1}, we conclude that $\mathbf{P}^{\Pi_1}\mathbf{c}_2 =\mathbf{0}$.
\item $i=3, i=4$ \\
According to $\mathbf{M}$ in \eqref{M for example 1}, $c_{3,1}=c_{3,3}=c_{4,2}=c_{4,4}\coloneqq c^*_{\{3,4\}}$, and $c_{3,2}=c_{3,4}=c_{4,1}=c_{4,3}=\mathbf{0}$. Therefore,
\begin{align}\label{p1 c3}
    \mathbf{P}\mathbf{c}_3 =\mathbf{P}_{\mathbf{v}_1}c_{3,1}+\mathbf{P}_{\mathbf{v}_2}c_{3,2}+\mathbf{P}_{\mathbf{v}_3}c_{3,3}+\mathbf{P}_{\mathbf{v}_4}c_{3,4}=(\mathbf{P}_{\mathbf{v}_1}+\mathbf{P}_{\mathbf{v}_3})c^*_{\{3,4\}} = \mathbf{0},
\end{align}
and
\begin{align}\label{p1 c4}
    \mathbf{P}\mathbf{c}_4 =\mathbf{P}_{\mathbf{v}_1}c_{4,1}+\mathbf{P}_{\mathbf{v}_2}c_{4,2}+\mathbf{P}_{\mathbf{v}_3}c_{4,3}+\mathbf{P}_{\mathbf{v}_4}c_{4,4}=(\mathbf{P}_{\mathbf{v}_2}+\mathbf{P}_{\mathbf{v}_4})c^*_{\{3,4\}} = \mathbf{0}.
\end{align}
As a result, 
\begin{align*}
    \mathbf{P}^{\Pi_1}\mathbf{c}_3 = \mathbf{P}_{\mathbf{v}_4}c_{3,1}+\mathbf{P}_{\mathbf{v}_3}c_{3,2}+\mathbf{P}_{\mathbf{v}_2}c_{3,3}+\mathbf{P}_{\mathbf{v}_1}c_{3,4} = (\mathbf{P}_{\mathbf{v}_4}+\mathbf{P}_{\mathbf{v}_2})c^*_{\{3,4\}},
\end{align*}
and due to \eqref{p1 c4}, we conclude that $\mathbf{P}^{\Pi_1}\mathbf{c}_3=\mathbf{0}$. Similarly,
\begin{align*}
    \mathbf{P}^{\Pi_1}\mathbf{c}_4 = \mathbf{P}_{\mathbf{v}_4}c_{4,1}+\mathbf{P}_{\mathbf{v}_3}c_{4,2}+\mathbf{P}_{\mathbf{v}_2}c_{4,3}+\mathbf{P}_{\mathbf{v}_1}c_{4,4} = (\mathbf{P}_{\mathbf{v}_3}+\mathbf{P}_{\mathbf{v}_1})c^*_{\{3,4\}},
\end{align*}
and due to \eqref{p1 c3}, we conclude that $\mathbf{P}^{\Pi_1}\mathbf{c}_4=\mathbf{0}$.
\item $i=5$\\
Since $c_{5,1}=c_{5,2}=c_{5,3}=c_{5,4}\coloneqq c^*_{\{5\}}$ in the fifth column of $\mathbf{M}$ according to \eqref{M for example 1}, 
\begin{align*}
    \mathbf{P}^{\Pi_1}\mathbf{c}_5 =
    \mathbf{P}_{\mathbf{v}_4}c_{5,1}+\mathbf{P}_{\mathbf{v}_3}c_{5,2}+\mathbf{P}_{\mathbf{v}_2}c_{5,3}+\mathbf{P}_{\mathbf{v}_1}c_{5,4} &= (\mathbf{P}_{\mathbf{v}_4} + \mathbf{P}_{\mathbf{v}_3} + \mathbf{P}_{\mathbf{v}_2} + \mathbf{P}_{\mathbf{v}_1})c^*_{\{5\}} \\
    &  = \mathbf{P}_{\mathbf{v}_1}c_{5,1}+\mathbf{P}_{\mathbf{v}_2}c_{5,2}+\mathbf{P}_{\mathbf{v}_3}c_{5,3}+\mathbf{P}_{\mathbf{v}_4}c_{5,4} = \mathbf{P}\mathbf{c}_5 = \mathbf{0}.
\end{align*}
\item $i=6, i=7, i=8, i=9$ \\
According to $\mathbf{M}$ in \eqref{M for example 1}, $c_{6,1}=c_{7,3}=c_{8,2}=c_{9,4}\coloneqq c^*_{\{6,7,8,9\}}$. Thus, 
\begin{align*}
    \mathbf{P}^{\Pi_1}\mathbf{c}_6 &= \mathbf{P}_{\mathbf{v}_4}c_{6,1}+\mathbf{P}_{\mathbf{v}_3}c_{6,2}+\mathbf{P}_{\mathbf{v}_2}c_{6,3}+\mathbf{P}_{\mathbf{v}_1}c_{6,4} = \mathbf{P}_{\mathbf{v}_4}c^*_{\{6,7,8,9\}} = \mathbf{P}\mathbf{c}_9 = \mathbf{0}, \\
    \mathbf{P}^{\Pi_1}\mathbf{c}_7 &= \mathbf{P}_{\mathbf{v}_4}c_{7,1}+\mathbf{P}_{\mathbf{v}_3}c_{7,2}+\mathbf{P}_{\mathbf{v}_2}c_{7,3}+\mathbf{P}_{\mathbf{v}_1}c_{7,4} = \mathbf{P}_{\mathbf{v}_2}c^*_{\{6,7,8,9\}} = \mathbf{P}\mathbf{c}_8 = \mathbf{0}, \\
    \mathbf{P}^{\Pi_1}\mathbf{c}_8 &= \mathbf{P}_{\mathbf{v}_4}c_{8,1}+\mathbf{P}_{\mathbf{v}_3}c_{8,2}+\mathbf{P}_{\mathbf{v}_2}c_{8,3}+\mathbf{P}_{\mathbf{v}_1}c_{8,4} = \mathbf{P}_{\mathbf{v}_3}c^*_{\{6,7,8,9\}} = \mathbf{P}\mathbf{c}_7 = \mathbf{0}, \\
    \mathbf{P}^{\Pi_1}\mathbf{c}_9 &= \mathbf{P}_{\mathbf{v}_4}c_{9,1}+\mathbf{P}_{\mathbf{v}_3}c_{9,2}+\mathbf{P}_{\mathbf{v}_2}c_{9,3}+\mathbf{P}_{\mathbf{v}_1}c_{9,4} = \mathbf{P}_{\mathbf{v}_1}c^*_{\{6,7,8,9\}} = \mathbf{P}\mathbf{c}_6 = \mathbf{0}.
\end{align*}
\item $i=10, i=11$\\
This case is exactly like $i=1, i=2$.
\item $i=12, i=13$\\
This case is exactly like $i=3, i=4$.
\end{itemize}
We proved that in Example \ref{example 1}, $\mathbf{P}^{\Pi_1}$ in \eqref{pi_1 in example 1} is a non-effective permutation of $\mathbf{P}$, i.e. $\mathbf{P}^{\Pi_1}\mathbf{M}=\mathbf{0}$.
\end{example}
\begin{remark}
Any non-effective permutation $\Pi:\{1,\dots,v^{\beta}\}\rightarrow\{1,\dots,v^{\beta}\}$ of the relation matrix, $\mathbf{P}^{\Pi}$, is in the row span of the relation matrix $\mathbf{P}$. Since $\mathbf{P}$ makes up the left null space of $\mathbf{M}$ by definition, $\mathbf{P}^{\Pi}\mathbf{M}=\mathbf{0}$ means $\mathbf{P}^{\Pi}$ should be in the row span of $\mathbf{P}$. In other words, any row of $\mathbf{P}^{\Pi}$ is a linear combination of the rows of $\mathbf{P}$.
\end{remark}
\begin{lemma}\label{number of non-effective permutations lemma}
In a $\textrm{Vers}(f,N,K,\beta,v,\textrm{Lagrange})$ system, there are $(v!)^{\beta}$ non-effective permutations of the relation matrix, out of all $(v^{\beta})!$ permutations.
\end{lemma}
We know that the relation matrix has $v^{\beta}$ sub-matrices and hence, there are $(v^{\beta})!$ possible permutations of relation matrix, including effective and non-effective. Moreover, we know that each sub-matrix is associated with a version vector that has $\beta$ elements from $[v]$. There are $v!$ permutations for each of those $\beta$ elements, and $(v!)^{\beta}$ permutations in total. This lemma states that only those $(v!)^{\beta}$ permutations are non-effective.
The proof of this lemma is in Appendix \ref{number of non-effective permutations proof}.

\section{Fundamental Limit of $\textrm{Vers}(f,N,K,\beta,v,\textrm{Lagrange})$ }\label{main result section}
In this section, we prove Theorem \ref{main result}, and provide achievability and converse proofs.
\subsection{Converse proof}\label{converse proof}
In this proof, we show that there exists a particular adversarial behavior, for which, $t=v^{\beta}d(K-1)$ messages from $t$ workers are not enough to find $f(X_k), k\in\mathcal{H}$ correctly. In other words, we show that $\{f(W_n),\mathsf{tag}_n,~n\in\mathcal{T}\}$, where $\mathcal{T}\subseteq [N]$, and $|\mathcal{T}|=v^{\beta}d(K-1)$, leads to more that one possible value for at least one $f(X_k), k\in\mathcal{H}$.

Consider the adversarial behavior where the $\beta$ adversarial data owners collude and distribute their messages to workers in $\mathcal{T}$ such that for every $\mathbf{v}\in\mathcal{V}$, there exist exactly $d(K-1)$ workers in $\mathcal{T}$ that receive adversarial messages whose versions are according to $\mathbf{v}$.

For a set $\mathcal{S}\subseteq\mathbb{F}$, and an integer $D\in\mathbb{N}$, let $\textrm{Van}_{\mathcal{S}}^D$ be a $|\mathcal{S}|\times(D+1)$ Vandermonde matrix that has $|\mathcal{S}|$ rows, and each row consists of powers $0,1,\dots,D$ of an element of $\mathcal{S}$. For example, 
\begin{align*}
    \textrm{Van}_{\{\alpha_1,\alpha_2\}}^3=\begin{bmatrix}
    \alpha_1^3 & \alpha_1^2 & \alpha_1 & 1 \\
    \alpha_2^3 & \alpha_2^2 & \alpha_2 & 1
    \end{bmatrix}.
\end{align*}
\textbf{Step $1$}. Suppose that an honest data owner $k\in\mathcal{H}$ receives $y_n=f(W_n)$, and $\mathsf{tag}_n$, from worker $n\in\mathcal{T}$. The honest data owner uses tags from workers in $\mathcal{T}$ to make up $v^{\beta}$ disjoint sets $\mathcal{T}_1,\dots,\mathcal{T}_{v^{\beta}}\subseteq\mathcal{T}$, such that each set contains workers whose tags are equal, but different from tags of workers in other sets. According to the adversarial behavior that we consider in this proof, $N_i\coloneqq|\mathcal{T}_i|=d(K-1)$, $i\in[v^{\beta}]$.
The properties of the tag function ensure the honest data owner that the $\beta$ adversarial messages used by workers $n\in\mathcal{T}_i$ and $n'\in\mathcal{T}_j$ from different sets are different in version of at least one message. Recall that the honest data owner does not know the adversarial data owners. Moreover, the honest data owner cannot know which workers have previously received the same message from the adversarial data owners, and thus have used the same messages in their computations, or, which workers have previously received different messages from the adversarial data owners, and thus have used different messages in their computations. \\
\textbf{Step $2$}. Let $\mathbf{y}_i=[y_j]_{j\in\mathcal{T}_i}$, and $\mathbf{Q}_i=\textrm{Van}_{\{\alpha_n,n\in\mathcal{T}_i\}}^{d(K-1)}$, for $i\in[v^{\beta}]$. There is an underlying mapping $\phi:[v^{\beta}]\rightarrow[v^{\beta}]$, that indicates the adversarial messages that workers in $\mathcal{T}_1,\dots,\mathcal{T}_{v^{\beta}}$ have previously received from adversarial data owners. In particular, $\phi(i)=j$ means $\mathbf{y}_i=\mathbf{Q}_i\mathbf{u}_{\mathbf{v}_j}$, where $i,j\in[v^{\beta}]$ and $u_{\mathbf{v}_j}$, $\mathbf{v}_j\in\mathcal{V}$ is a coefficient sub-vector defined in \eqref{u (subvector) definition}. In other words, $\mathbf{y}_i=\mathbf{Q}_i\mathbf{u}_{\mathbf{v}_j}$ means that workers $n\in\mathcal{T}_i$ have used adversarial messages whose versions are in $\mathbf{v}_j$. Without loss of generality, assume that $\phi$ is the identity permutation, i.e. $\mathbf{y}_i=\mathbf{Q}_i\mathbf{u}_{\mathbf{v}_i}$, for all $i\in[v^{\beta}]$.

The honest data owner $k$ considers a coefficient vector $\mathbf{U}'$ as unknown, and tries to solve equations to find $\mathbf{U}'$. However, the honest data owner does not know $\phi$, so it needs to consider all different possible cases for $\phi$, and make sure that all cases result in a single $f(X_k)$. Suppose that the honest data owner considers
an effective permutation $\Pi:[v^{\beta}]\rightarrow [v^{\beta}]$, and assumes that $\mathbf{y}_i=\mathbf{Q}_i\mathbf{u}'_{\mathbf{v}_{\Pi(i)}}$. Therefore,
\begin{align}\label{Q(u-u')=0 label}
    \mathbf{Q}_i (\mathbf{u}_{\mathbf{v}_i} - \mathbf{u}'_{\mathbf{v}_{\Pi(i)}}) = \mathbf{0}, \quad i\in[v^{\beta}].
\end{align}
\textbf{Step $3$}. Recall that $f(X_k)=f(q^{(\mathbf{v})}(\omega_k))$, $k\in\mathcal{H}$, for any $\mathbf{v}\in\mathcal{V}$, so
\begin{align}
    \left[ f(X_k)\right]_{k\in\mathcal{H}} =  \big(\textrm{Van}_{\{\omega_k,k\in\mathcal{H}\}}^{d(K-1)+1}\big)\mathbf{u}_{\mathbf{v}_i}, \quad i\in[v^{\beta}]
\end{align}
Consequently, all values of $\textrm{Van}_{\{\omega_k,k\in\mathcal{H}\}}^{d(K-1)+1}\mathbf{u}_{\mathbf{v}_i}$ for $i\in[v^{\beta}]$ are equal. Similarly, all values of $\textrm{Van}_{\{\omega_k,k\in\mathcal{H}\}}^{d(K-1)+1}\mathbf{u}'_{\mathbf{v}_{\Pi(i)}}$, for $i\in[v^{\beta}]$ are equal as well. We will show that any effective permutation $\Pi$ results in different values for $\textrm{Van}_{\{\omega_k,k\in\mathcal{H}\}}^{d(K-1)+1}\mathbf{u}_{\mathbf{v}_i}$ and $\textrm{Van}_{\{\omega_k,k\in\mathcal{H}\}}^{d(K-1)+1}\mathbf{u}'_{\mathbf{v}_{\Pi(i)}}$, $i,j \in [v^{\beta}]$. If there exists one $i^*\in[v^{\beta}]$ such that 
\begin{align}
    \textrm{Van}_{\{\omega_k,k\in\mathcal{H}\}}^{d(K-1)+1} (\mathbf{u}_{\mathbf{v}_{i^*}} - \mathbf{u}'_{\mathbf{v}_{\Pi(i^*)}}) = \mathbf{0},
\end{align}
then we can conclude that 
\begin{align}
    \left[ f(X_k)\right]_{k\in\mathcal{H}} =  \big(\textrm{Van}_{\{\omega_k,k\in\mathcal{H}\}}^{d(K-1)+1}\big)\mathbf{u}'_{\mathbf{v}_{\Pi(i)}}, \quad i\in [v^{\beta}],
\end{align}
because
\begin{align}
    \left[ f(X_k)\right]_{k\in\mathcal{H}} = \big(\textrm{Van}_{\{\omega_k,k\in\mathcal{H}\}}^{d(K-1)+1}\big)\mathbf{u}_{\mathbf{v}_{i^*}}
\end{align}
and
\begin{align}
    \big(\textrm{Van}_{\{\omega_k,k\in\mathcal{H}\}}^{d(K-1)+1}\big)\mathbf{u}'_{\mathbf{v}_{\Pi(i)}} = \big(\textrm{Van}_{\{\omega_k,k\in\mathcal{H}\}}^{d(K-1)+1}\big)\mathbf{u}'_{\mathbf{v}_{\Pi(i^*)}}, \quad i\in [v^{\beta}].
\end{align}
In other words, if there exists such $i^*$, permutation $\Pi$ results in correct answer for $f(X_k), k\in\mathcal{H}$. By contradiction, we suppose that there exists such $i^*$. Therefore, 
\begin{align}\label{Q(u-u')=0 and Van(u-u')=0}
   \begin{split}
        \mathbf{Q}_i (\mathbf{u}_{\mathbf{v}_i} - \mathbf{u}'_{\mathbf{v}_{\Pi(i)}}) &= \mathbf{0}, \quad i\in[v^{\beta}],\\
    \big(\textrm{Van}_{\{\omega_k,k\in\mathcal{H}\}}^{d(K-1)+1}\big) (\mathbf{u}_{\mathbf{v}_i} - \mathbf{u}'_{\mathbf{v}_{\Pi(i)}}) &= \mathbf{0}, \quad i\in[v^{\beta}].
   \end{split}
\end{align}
The matrix $\mathbf{Q}_i$ only contains powers of $\alpha_n, n\in\mathcal{T}_i$, and $\textrm{Van}_{\{\omega_k,k\in\mathcal{H}\}}^{d(K-1)+1}$ contains powers of $\omega_k,k\in\mathcal{H}$, and these elements are distinct and chosen uniformly at random from $\mathbb{F}$. Therefore, the matrix 
\begin{align*}
\left[\begin{array}{c}
\mathbf{Q}_i \\\hdashline  \textrm{Van}_{\{\omega_k,k\in\mathcal{H}\}}^{d(K-1)+1}
\end{array}\right]
\end{align*}
is MDS, its dimension is $(d(K-1)+h)\times (d(K-1)+1)$, and its rank is $d(K-1)+1$. The length of $\mathbf{u}_{\mathbf{v}_i} - \mathbf{u}'_{\mathbf{v}_{\Pi(i)}}$ is $d(K-1)+1$, therefore any $d(K-1)+1$ equations of \eqref{Q(u-u')=0 and Van(u-u')=0} result in
\begin{align}\label{result of contradiction}
    \mathbf{u}_{\mathbf{v}_i} = \mathbf{u}'_{\mathbf{v}_{\Pi(i)}},\quad i\in[v^{\beta}].
\end{align}
\textbf{Step $4$}. We know that $\mathbf{P}\mathbf{U}=0$ and $\mathbf{P}\mathbf{U}'=0$, thus,
\begin{align}
    \mathbf{P}_{\mathbf{v}_1}\mathbf{u}_{\mathbf{v}_1}+\dots+\mathbf{P}_{\mathbf{v}_{v^{\beta}}}\mathbf{u}_{\mathbf{v}_{v^{\beta}}} &= \mathbf{0}, \label{PU=0} \\
    \mathbf{P}_{\mathbf{v}_{\Pi(1)}}\mathbf{u}'_{\mathbf{v}_{\Pi(1)}}+\dots+\mathbf{P}_{\mathbf{v}_{\Pi(v^{\beta})}}\mathbf{u}'_{\mathbf{v}_{\Pi(v^{\beta})}} &= \mathbf{0}. \label{PU'=0}
\end{align}
We substitute \eqref{result of contradiction} in \eqref{PU'=0}, 
\begin{align}
    \mathbf{P}_{\mathbf{v}_{\Pi(1)}}\mathbf{u}_{\mathbf{v}_1}+\dots+\mathbf{P}_{\mathbf{v}_{\Pi(v^{\beta})}}\mathbf{u}_{\mathbf{v}_{v^{\beta}}} &= \mathbf{0}, \label{PU'=0 2}
\end{align}
\begin{comment}
thus,
\begin{align}\label{stacked P}
    \begin{bmatrix}
    \mathbf{P}_1 & \dots & \mathbf{P}_{v^{\beta}} \\
    \mathbf{P}_{\Pi_1} & \dots & \mathbf{P}_{\Pi_{v^{\beta}}}
    \end{bmatrix} \mathbf{U} = \mathbf{0}
\end{align}
\end{comment}
Recall that $\Pi$ is an effective permutation, so according to  definition, $[\mathbf{P}_{\mathbf{v}_{\Pi(1)}}~\dots~\mathbf{P}_{\mathbf{v}_{\Pi(v^{\beta})}}]$ is not in row span of $\mathbf{P}$. Therefore, equation \eqref{PU'=0 2} exerts additional constraint on $\mathbf{U}$ and consequently all messages of data owners. However, messages of data owners, honest and adversarial, can have any value and any constraint on them means contradiction. This concludes the converse proof.

\subsection{Achievability Proof}\label{achievable proof}
In the achievability proof, we suppose that the honest data owner $k\in\mathcal{H}$ receives $t^*=v^{\beta}d(K-1)+1$ messages from $t^*$ workers. The honest data owner uses the received tags to make up $v^{\beta}$ disjoint sets of workers, such that each set contains workers whose tags are equal but different from tags of workers in other sets. Since $t^*>v^{\beta}d(K-1)$, one of these sets definitely contains at least $d(K-1)+1$ workers, which we call $\mathcal{T^*}$. There exists $i^*\in[v^{\beta}]$ such that
\begin{align}
    \big(\textrm{Van}_{\{\alpha_n,n\in\mathcal{T^*}\}}^{d(K-1)}\big)\mathbf{u}_{\mathbf{v}_{i^*}} = [y_j]_{j\in\mathcal{T^*}}.
\end{align}
This equation determines $\mathbf{u}_{\mathbf{v}_{i^*}}$ uniquely. Then, the honest data owner can find $f(X_k)=\big(\textrm{Van}_{\omega_k}^{d(K-1)}\big)\mathbf{u}_{\mathbf{v}_{i^*}}$. This completes the proof.

One of the assumptions in Theorem \ref{main result} is that workers are honest and send the correct values of $f(W_n),\mathsf{tag}_n,~n\in[N]$ to data owners. Suppose that $a\in\mathbb{N}, a\le N$ workers are malicious and may send incorrect $f(W_n)$ or $\mathsf{tag}_n$ back to data owners. Incorrect tag values are detrimental, because they mislead data owners in partitioning. Following the notations in converse proof, assume that $\mathbf{y}_i=\mathbf{Q}_i\mathbf{u}_{\mathbf{v}_i}+\mathbf{e}_i$  is the truth, and $\mathbf{y}_i=\mathbf{Q}_i\mathbf{u}'_{\Pi(\mathbf{v}_i)}+\mathbf{e}'_i$ is an assumption that an honest data owner makes because of not knowing the truth, and $\mathbf{e}_i, \mathbf{e}'_i$ are error vectors. Therefore, in the presence of adversarial workers, equation \eqref{Q(u-u')=0 label} does not hold anymore, and the approach used in the previous converse proof can not be used here. However, we know from coding theory that results from $t=v^{\beta}d(K-1)+1+2a$ workers are enough to decode $f(X_k),k\in\mathcal{H}$ correctly.

\begin{appendices}

\section{Proof of Theorem \ref{random tag function theory}}
In order to show that a tag function exists, we consider all $J$ functions from $\mathbb{U}^K$ to $\mathbb{U}$ to be equiprobable and calculate the average probability of the event that two arbitrary workers produce the same tag. Consider two workers $n_1,n_2\in[N]$ that receive different $X_{\mathcal{A},n_1}$ and $X_{\mathcal{A},n_2}$ from adversarial data owners, and $X_{\mathcal{H}}$ from honest data owners. We show that the average probability of $\mathsf{tag}_{n_1}=\mathsf{tag}_{n_2}$, taken over all $J$ functions, and the messages of honest data owners, is very small. Then, we conclude that there exists a function, for which the average probability of $\mathsf{tag}_{n_1}=\mathsf{tag}_{n_2}$, taken over  the messages of honest data owners, is very small, which means a tag function exists.

Any function $J$ from $\mathbb{U}^K$ to $\mathbb{U}$, i.e. from a large space to a small space, would map some different inputs to the same outputs. For a function $J$, and two adversarial messages $X_{\mathcal{A},n_1}$ and $X_{\mathcal{A},n_2}$, we define a function $g$ such that $\mathcal{S} \coloneqq  g(J,X_{\mathcal{A},n_1},X_{\mathcal{A},n_2})\subseteq \mathbb{U}^h$ is the set that
\begin{align*}
    X_{\mathcal{H}} \in \mathcal{S} : J(X_{\mathcal{H}}, X_{\mathcal{A},n_1}) = J(X_{\mathcal{H}}, X_{\mathcal{A}, n_2}), \\
    X_{\mathcal{H}} \in \mathbb{U}^h\setminus\mathcal{S} : J(X_{\mathcal{H}}, X_{\mathcal{A}, n_1}) \neq J(X_{\mathcal{H}}, X_{\mathcal{A}, n_2}).
\end{align*}
Since the data of honest data owners are uniformly and independently chosen from $\mathbb{U}^h$, and are independent of $X_{\mathcal{A},n_1}$ and $X_{\mathcal{A},n_2}$,
\begin{align}\label{first}
    \textrm{Pr}\big( J(X_{\mathcal{H}},X_{\mathcal{A},n_1}) = J(X_{\mathcal{H}},X_{\mathcal{A},n_2})\big) = \frac{|\mathcal{S}|}{\mathbb{U}^h},
\end{align}
where $\mathcal{S} = g(J,X_{\mathcal{A},n_1},X_{\mathcal{A},n_2})\subseteq \mathbb{U}^h$. Since all $J$ functions are equiprobable, i.e. $\textrm{Pr}(J) = \frac{1}{|\mathbb{U}|^K}$, 
\begin{align}\label{second}
    \sum_{J:g(J,X_{\mathcal{A},n_1},X_{\mathcal{A},n_2})=\mathcal{S}} \textrm{Pr}(J) = (1-\frac{1}{|\mathbb{U}|})^{(|\mathbb{U}|^h-|\mathcal{S}|)}(\frac{1}{|\mathbb{U}|})^{|\mathcal{S}|}.
\end{align}
Using \eqref{first} and \eqref{second},
\begin{align}
    \textrm{Pr}(\mathsf{tag}_{n_1}=\mathsf{tag}_{n_2}) &= \sum_{\mathcal{S}}\sum_{J:g(J,X_{\mathcal{A},n_1},X_{\mathcal{A},n_2})=\mathcal{S}} \textrm{Pr}(J) \textrm{Pr}\big( J(X_{\mathcal{H}},X_{\mathcal{A},n_1}) = J(X_{\mathcal{H}},X_{\mathcal{A},n_2})\big) \nonumber\\
    &= \sum_{\mathcal{S}} (1-\frac{1}{|\mathbb{U}|})^{(|\mathbb{U}|^h-|\mathcal{S}|)}(\frac{1}{|\mathbb{U}|})^{|\mathcal{S}|} \frac{|\mathcal{S}|}{\mathbb{U}^h} \nonumber\\
    &= \sum_{s=1}^{|\mathbb{U}|^h}   \binom{|\mathbb{U}|^h}{s} (1-\frac{1}{|\mathbb{U}|})^{(|\mathbb{U}|^h-s)}(\frac{1}{|\mathbb{U}|})^s \frac{s}{\mathbb{U}^h} \nonumber\\
    &= \sum_{s=1}^{|\mathbb{U}|^h}   \binom{|\mathbb{U}|^h-1}{s-1} (1-\frac{1}{|\mathbb{U}|})^{(|\mathbb{U}|^h-s)}(\frac{1}{|\mathbb{U}|})^s = \frac{1}{|\mathbb{U}|}
\end{align}
If $\mathbb{U}$ is large enough that $\frac{1}{\epsilon}\le |\mathbb{U}|$, then $\textrm{Pr}(\mathsf{tag}_{n_1}=\mathsf{tag}_{n_2})\le \epsilon$, and the proof is complete.

\section{Proof of Lemma \ref{number of non-effective permutations lemma}}\label{number of non-effective permutations proof}
In this appendix, we prove that there are $(v!)^{\beta}$ non-effective permutations of the relation matrix. In Section \ref{definitions section}, we examined an example, and proved the non-effectiveness of a permutation of its relation matrix. The understanding of that example is helpful for understanding the proof in this section.

We show that all permutations of the relation matrix of the form $\Pi = \pi_1\times\pi_2\times\dots\times\pi_{\beta}$ are non-effective, i.e. $\mathbf{P}^{\Pi}\mathbf{M}=\mathbf{0}$, where $\pi_i:[v]\rightarrow[v]$, $i\in[\beta]$ is a permutation on $[v]$. We parse the proof into the following steps.

\subsection*{Step 1}
Consider column $i\in[L]$, $\mathbf{c}_i=[c_{\mathbf{v}_1,i}^\intercal~\dots~c_{\mathbf{v}_{v^{\beta}},i}^\intercal]^\intercal$ of $\mathbf{M}$, where the length of each so called \textit{sub-column} $c_{\mathbf{v}_i},i\in[v^{\beta}]$ is $d(K-1)+1$, $\{\mathbf{v}_1,\dots,\mathbf{v}_{v^{\beta}}\}=\mathcal{V}$, and $L$ is the length of monomial vector $\mathbf{X}$. 

The sub-column $c_{\mathbf{v},i}$, $\mathbf{v}\in\mathcal{V}, i\in[L]$, contains the coefficients of monomial $\mathbf{X}[i]$ in coefficient sub-vector $\mathbf{u}_{\mathbf{v}}$. If coefficient sub-vector $\mathbf{u}_{\mathbf{v}}$ does not have the monomial $\mathbf{X}[i]$, then $c_{\mathbf{v},i}=\mathbf{0}$. For two coefficient sub-vectors $\mathbf{u}_{\mathbf{v}}$ and $\mathbf{u}_{\mathbf{v'}}$ that both have $\mathbf{X}[i]$, $c_{\mathbf{v},i}=c_{\mathbf{v'},i}\coloneqq c^*_i$, because these coefficient sub-vectors only differ in adversarial messages, and not the constant coefficients. Therefore, the nonzero sub-columns in a column of $\mathbf{M}$ are the same. This is evident in $\mathbf{M}$ of Example \ref{example 1} in \eqref{M for example 1}.

We define two neighbour columns of $\mathbf{M}$ as following. Consider two columns $i,i'\in[L]$, where the adversarial and honest messages in monomials $\mathbf{X}[i]$ and $\mathbf{X}[i']$ are the same, but the versions of adversarial messages in them are different. For example, $i=6,i'=7$ in \eqref{M for example 1} correspond to monomials $X_1^{(1)}X_2^{(1)}$ and $X_1^{(2)}X_2^{(1)}$, which differ in the version of $X_1$, but both are made of the same messages $X_1$ and $X_2$.
For two neighbour columns, $c^*_i=c^*_{i'}$, i.e. the nonzero sub-columns in two neighbour columns are the same.
% Note that in for any coefficient sub-vector, at most one of neighbour columns is nonzero, because two versions of the same adversarial messages does not exist in a coefficient sub-vector.
For example, the nonzero sub-columns of columns $6$ and $7$ in \eqref{M for example 1} are the same.

\subsection*{Step 2}
Let $\mathcal{V}_{\mathbf{c}_i}$, $i\in[L]$, be the set of version vectors of coefficient sub-vectors that have the monomial $\mathbf{X}[i]$, or in other words, version vectors that correspond to nonzero sub-columns of $\mathbf{c}_i$. For instance, in Example \ref{example 1}, using \eqref{coefficient vector example} and \eqref{M for example 1}, we find that $\mathcal{V}_{\mathbf{c}_1}=\{[1,1,0],[1,2,0]\}$, $\mathcal{V}_{\mathbf{c}_2}=\{[2,1,0],[2,2,0]\}$, 
$\mathcal{V}_{\mathbf{c}_3}=\{[1,1,0],[2,1,0]\}$, and so forth. From the previous step, we know that $c_{\mathbf{v},i}=c_i^*$ for $\mathbf{v}\in\mathcal{V}_{\mathbf{c}_i}$.
For each column $\mathbf{c}_i$ of $\mathbf{M}$, there are two possible cases.
\begin{itemize}
    \item \textbf{Case 1}. The monomial $\mathbf{X}[i]$ is comprised of honest messages only, e.g. $(X_3)^2$. In this case, all sub-columns of $\mathbf{c}_i$ are equal and nonzero, i.e. $c_{\mathbf{v}_1,i} =\dots=c_{\mathbf{v}_{v^{\beta}},i}\coloneqq c^*_{i}$, and $\mathcal{V}_{\mathbf{c}_i}=\mathcal{V}$. In Example \ref{example 1}, $\mathbf{c}_5$ of $\mathbf{M}$ in \eqref{M for example 1} is of this kind, as
    \begin{align}
        c_{\mathbf{v}_1,5} = c_{\mathbf{v}_2,5} = c_{\mathbf{v}_3,5} = c_{\mathbf{v}_4,5} = \begin{bmatrix}
        \frac{1}{4} \\ -\frac{3}{2} \\ \frac{13}{4} \\ -3 \\ 1
        \end{bmatrix}
    \end{align}
    \item \textbf{Case 2}.
    The monomial $\mathbf{X}[i]$ contains adversarial messages, e.g. $X[i]=(X^{(1)}_1)^2$. In this case, version vectors $\mathbf{v}\in\mathcal{V}_{\mathbf{c}_i}$ have the same version value for the adversarial messages in $X[i]$. For example, $\mathcal{V}_{\mathbf{c}_1}=\{[1,1,0],[1,2,0]\}$, and version vectors in $\mathcal{V}_{\mathbf{c}_1}$ indicate version $(1)$ for $X_1$ because $X[i]=(X^{(1)}_1)^2$. As another example, for $X[11]=X_1^{(2)}X_3$ in \eqref{M for example 1}, $\mathcal{V}_{\mathbf{c}_{11}}=\{[2,1,0],[2,2,0]\}$, and version vectors in $\mathcal{V}_{\mathbf{c}_{11}}$ indicate version $(2)$ for $X_1$.
\end{itemize}

\subsection*{Step 3}
We know $\mathbf{P}\mathbf{c}_i=\mathbf{0}$, $i\in[L]$. Recall that $\mathbf{P}=[\mathbf{P}_{\mathbf{v}_1}\dots \mathbf{P}_{\mathbf{v}_{v^{\beta}}}]$, where $\{\mathbf{v}_1,\dots,\mathbf{v}_{v^{\beta}}\}=\mathcal{V}$.
Therefore,
\begin{align}
    \mathbf{P}\mathbf{c}_i = \sum_{\mathbf{v}\in\mathcal{V}_{\mathbf{c}_i}} \mathbf{P}_{\mathbf{v}}c_{\mathbf{v},i} = \bigg(\sum_{\mathbf{v}\in\mathcal{V}_{\mathbf{c}_i}} \mathbf{P}_{\mathbf{v}}\bigg)c^*_i= \mathbf{0}
\end{align}
Consider a permutation of form $\Pi = \pi_1\times\pi_2\times\dots\times\pi_{\beta}$, where $\pi_i:[v]\rightarrow[v]$, $i\in[\beta]$. Assume that this permutation is as following.
\begin{itemize}
    \item For a particular adversarial message $X_k, k\in\mathcal{A}$, and two particular version values $m,m'\in[v]$,
    \begin{align*}
        \pi_k[m] &= m', \\
        \pi_k[m'] &= m, \\
        \pi_k[l] &= l, l\neq m,m'.
    \end{align*}
    \item For all $j\in\mathcal{A}, j\neq k$,
    \begin{align*}
        \pi_j[l]=l, \quad l\in[v].
    \end{align*}
\end{itemize}
It suffices to prove the non-effectiveness of this permutation, because any other permutation is built from units of permutations like this. For each column $\mathbf{c}_i$, there are two possible cases following the cases in the previous step.
\begin{itemize}
    \item \textbf{Case 1}. If $\mathcal{V}_{\mathbf{c}_i}=\mathcal{V}$, then 
    \begin{align}
        \mathbf{P}^{\Pi}\mathbf{c}_i=(\sum_{\mathbf{v}\in\mathcal{V}} \mathbf{P}_{\mathbf{v}})c^*_{i}=\mathbf{P}\mathbf{c}_i=\mathbf{0}.
    \end{align}
    \item \textbf{Case 2}. Assume that monomial $\mathbf{X}[i]$ contains version $m$ of the adversarial message $k$. There is a monomial $\mathbf{X}[i']$, $i'\in[L]$ that differs from $\mathbf{X}[i]$ only in the version of $X_k$, which is $m'$ in $\mathbf{X}[i']$. Clearly,
    \begin{align}\label{Pc'_i}
        \mathbf{P}\mathbf{c}_{i'} = \sum_{\mathbf{v}\in\mathcal{V}_{\mathbf{c}_{i'}}} \mathbf{P}_{\mathbf{v}}c_{\mathbf{v},{i'}} = \bigg(\sum_{\mathbf{v}\in\mathcal{V}_{\mathbf{c}_{i'}}} \mathbf{P}_{\mathbf{v}}\bigg)c^*_{i'}= \mathbf{0}.
    \end{align}
     It is easy to verify that for each $\mathbf{v}'\in\mathcal{V}_{\mathbf{c}_{i'}}$, there is a $\mathbf{v}\in\mathcal{V}_{\mathbf{c}_i}$, where $\mathbf{v}'[k]=m'$ and $\mathbf{v}[k]=m$.
     From the definition in Step 1, we know that columns $i$ and $i'$ are neighbours and thus, $c^*_{i}=c^*_{i'}$. Therefore,
    \begin{align}
        \mathbf{P}^{\Pi}\mathbf{c}_i= \sum_{\substack{ \mathbf{v'}\in\mathcal{V}_{\mathbf{c}_{i'}} \\ \mathbf{v}\in\mathcal{V}_{\mathbf{c}_i}}} \mathbf{P}_{\mathbf{v'}}c_{\mathbf{v},{i}} = \bigg(\sum_{\mathbf{v}\in\mathcal{V}_{\mathbf{c}_{i'}}} \mathbf{P}_{\mathbf{v}}\bigg)c^*_{i}= \bigg(\sum_{\mathbf{v}\in\mathcal{V}_{\mathbf{c}_{i'}}} \mathbf{P}_{\mathbf{v}}\bigg)c^*_{i'} = \mathbf{0},
    \end{align}
    where the last equality is due to \eqref{Pc'_i}. In summary, we proved $\mathbf{P}^{\Pi}\mathbf{c}_i=\mathbf{0}$ using $\mathbf{P}\mathbf{c}_{i'} = \mathbf{0}$.
\end{itemize}

\end{appendices} 
    
\bibliographystyle{IEEEtran}
\bibliography{ref}
\end{document}